\documentclass[preprint,12pt]{elsarticle}
\usepackage{graphicx}
\usepackage{amsbsy}
\usepackage{amsmath}
\usepackage{amssymb}
\usepackage{amsmath,bm}
\usepackage{latexsym}
\usepackage{indentfirst}
\usepackage{epsfig,float}
\usepackage{wrapfig,lipsum}
\usepackage{mathrsfs}
\usepackage{makecell}
\usepackage{algorithm}
\usepackage{algorithmic}
\usepackage[algo2e,ruled,vlined]{algorithm2e}
\newtheorem{definition}{Definition}

\newtheorem{lem}{Lemma}
\newtheorem{corollary}{Corollary}
\newtheorem{clm}{Claim}[lem]
\newtheorem{remark}{Remark}

\newtheorem{proposition}{Proposition}

\usepackage{thm-restate}

\newcommand{\FPT}{$\mathsf{FPT}$ }

\usepackage{graphicx}
  \baselineskip 8 mm

\begin{document}
\begin{frontmatter}

\title{Parameterized Results on Acyclic Matchings
with Implications for Related Problems}

 \author{Juhi Chaudhary} \ead{juhic@post.bgu.ac.il}
 \author{Meirav Zehavi}
 \ead{meiravze@bgu.ac.il}
 
 \address{Department of Computer Science, Ben-Gurion University of the Negev, Beersheba, Israel}

\begin{abstract}
 A matching $M$ in a graph $G$ is an \emph{acyclic matching} if the subgraph of $G$ induced by the endpoints of the edges of $M$ is a forest. Given a graph $G$ and a positive integer $\ell$, \textsc{Acyclic Matching} asks whether $G$ has an acyclic matching of size (i.e., the number of edges) at least $\ell$. In this paper, we first prove that assuming $\mathsf{W[1]\nsubseteq FPT}$, there does not exist any $\mathsf{FPT}$-approximation algorithm for \textsc{Acyclic Matching} that approximates it within a constant factor when the parameter is the size of the matching. Our reduction is general in the sense that it also asserts $\mathsf{FPT}$-inapproximability for \textsc{Induced Matching} and \textsc{Uniquely Restricted Matching} as well. We also consider three below-guarantee parameters for \textsc{Acyclic Matching}, viz. $\frac{n}{2}-\ell$, $\mathsf{MM(G)}-\ell$, and $\mathsf{IS(G)}-\ell$, where $n$ is the number of vertices in $G$, $\mathsf{MM(G)}$ is the matching number of $G$, and $\mathsf{IS(G)}$ is the independence number of $G$. Furthermore, we show that \textsc{Acyclic Matching} does not exhibit a polynomial kernel with respect to vertex cover number (or vertex deletion distance to clique) plus the size of the matching unless $\mathsf{NP}\subseteq \mathsf{coNP} \slash \mathsf{poly}$.
\end{abstract}

\date{}
\end{frontmatter}
\date{nodate}
 \section{Introduction}\label{intro}

 Matchings form a central topic in Graph Theory and Combinatorial Optimization \cite{lovaz}. In addition to their theoretical fruitfulness, matchings have various
practical applications, such as assigning new physicians to hospitals, students to high schools, clients to server clusters, kidney donors to recipients \cite{manlove}, and so on. Moreover, matchings can be associated with the concept of \emph{edge colorings} \cite{baste2020approximating,basteurm,vizing}, and they are a useful tool for finding optimal solutions or bounds in competitive optimization games on graphs \cite{bachstein,goddard1}. Matchings have also found applications in the area of Parameterized Complexity and Approximation Algorithms. For example, \emph{Crown Decomposition}, a construction widely used for kernelization, is based on classical matching theorems \cite{cygan}, and matchings are one of the oldest tools in designing approximation algorithms for graph problems (e.g., consider the classical 2-approximation algorithm for \textsc{Vertex Cover} \cite{ausiello2012complexity}).

A matching $M$ is a \emph{$\mathcal{P}$ matching} if $G[V_{M}]$ (the subgraph induced by the endpoints of edges in $M$) has the property $\mathcal{P}$, where $\mathcal{P}$ is some graph property. The problem of deciding whether a graph admits a $\mathcal{P}$ matching of a given size (number of edges) has been investigated for
several graph properties \cite{francis, goddard, golumbic, panda1, panda, stockmeyer}. If the property $\mathcal{P}$ is that of being a graph, a disjoint union of edges, a forest, or having a unique perfect matching, then a $\mathcal{P}$ matching is a \emph{matching} \cite{micali}, an \emph{induced matching} \cite{stockmeyer}, an \emph{acyclic matching} \cite{goddard}, and a \emph{uniquely restricted matching} \cite{golumbic}, respectively. In this paper, we focus on the case of acyclic matchings and also discuss implications for the cases of induced matchings and uniquely restricted matchings. In particular, we study the parameterized complexity of these problems.

\subsection{Problem Definitions and Related Works}\label{related}

Given a graph $G$ and a positive integer $\ell$, \textsc{Acyclic Matching} asks whether $G$ has an acyclic matching of size at least $\ell$. Goddard et al. \cite{goddard} introduced the concept of acyclic matching, and since then, it has gained significant popularity in the literature \cite{baste2020approximating,baste2018degenerate,furst2019some,panda1,panda}. In what follows, we present a brief survey of the algorithmic results concerning acyclic matching. \textsc{Acyclic Matching} is known to be $\mathsf{NP}$-$\mathsf{complete}$ for perfect elimination bipartite graphs, a subclass
of bipartite graphs \cite{panda}, star-convex bipartite graphs \cite{panda1}, and dually chordal graphs \cite{panda1}. On the positive side, \textsc{Acyclic Matching} is known to be polynomial-time solvable for chordal graphs \cite{baste2018degenerate} and bipartite permutation graphs \cite{panda}. F\"urst and Rautenbach \cite{furst2019some} showed that it is $\mathsf{NP}$-$\mathsf{hard}$ to decide whether
a given bipartite graph of maximum degree at most $4$ has a maximum matching
that is acyclic. Furthermore, they characterized the graphs for which every maximum matching is acyclic and gave linear-time algorithms to compute a maximum acyclic matching in graph classes such as $P_{4}$-free graphs and $2P_{3}$-free graphs. With respect to approximation, Panda and Chaudhary \cite{panda1} showed that \textsc{Acyclic Matching} is hard to approximate within factor $n^{1-\epsilon}$ for every $\epsilon>0$ unless $\mathsf{P=NP}$. Apart from that, Baste et al. \cite{baste2018degenerate} showed that finding a maximum cardinality 1-degenerate matching in a graph $G$ is
equivalent to finding a maximum acyclic matching in $G$. 

From the viewpoint of Parameterized Complexity, Hajebi and Javadi \cite{hajebi} showed that \textsc{Acyclic Matching} is fixed-parameter tractable ($\mathsf{FPT}$) when parameterized by treewidth using Courcelle's theorem. Furthermore, they showed that the problem is $\mathsf{W[1]}$-$\mathsf{hard}$ on bipartite graphs when parameterized by the size of the matching. 
However, under the same parameter, the authors showed that the problem is $\mathsf{FPT}$ for line graphs, $C_{4}$-free graphs, and every proper minor-closed class of graphs. Additionally, they showed that \textsc{Acyclic Matching} is \FPT when parameterized by the size of the matching plus the number of cycles of length four in the given graph. See Table \ref{tablepig} for a tabular view of known parameterized results concerning acyclic matchings.
\begin{table} [t]
	\scalebox{0.90}{

	\begin{tabular}{|l|l|l|}
		\hline

	&	\underline{\textbf{Parameter}} & \underline{\textbf{Result}} \\
 
	1. &	size of the matching & $\mathsf{W[1]}$-$\mathsf{hard}$ on bipartite graphs \cite{hajebi}  \\
		
	2. &	size of the matching & $\mathsf{FPT}$ on line graphs \cite{hajebi} \\
		
	3. & size of the matching & $\mathsf{FPT}$ on every proper minor-closed class of graphs \cite{hajebi} \\  
 
    4. & size of the matching plus $c_{4}$ & $\mathsf{FPT}$ \cite{hajebi} \\
 
	5. & $\mathsf{tw}$ & $\mathsf{FPT}$ \cite{hajebi} \\
 
		\hline
	\end{tabular}

 }
	
		\caption{Overview of parameterized results for \textsc{Acyclic Matching}. Here, $c_4$ and $\mathsf{tw}$ denote the number of cycles with length four and the treewidth of
the input graph, respectively.}
	\label{tablepig}

\end{table}

Recall that a matching $M$ is an induced matching if $G[V_{M}]$ is a disjoint union of $K_{2}'s$. In fact, note that every induced matching is an acyclic matching, but the converse need not be true. Given a graph $G$ and a positive integer $\ell$, \textsc{Induced Matching} asks whether $G$ has an induced matching of size at least $\ell$. Stockmeyer and Vazirani introduced the concept of induced matching as the
“risk-free” marriage problem in 1982 \cite{stockmeyer}. Since then, induced matchings have been studied extensively due to its various applications and relations with other graph problems \cite{cameron, cooley, erman, kanj, klemz, Ko, kowalik, moser, panda0, stockmeyer, zito}. Given a graph $G$ and a positive integer $\ell$, \textsc{Induced Matching Below Triviality ($\mathsf{IMBT}$)} asks whether $G$ has an induced matching of size at least $\ell$ with the parameter $k=\frac{n}{2}-\ell$, where $n=|V(G)|$. $\mathsf{IMBT}$ has been studied in the literature, albeit under different names:  Moser and  Thilikos \cite{moser1} gave an algorithm to solve $\mathsf{IMBT}$ in $\mathcal{O}(9^k\cdot n^{\mathcal{O}(1)})$ time. Subsequently, Xiao and Kou \cite{xiao} developed an algorithm running in $\mathcal{O}(3.1845^k\cdot n^{\mathcal{O}(1)})$ time.
\textsc{Induced Matching} with respect to some new \emph{below guarantee} parameterizations have also been studied recently \cite{koana}.

Next, recall that a matching $M$ is a uniquely restricted matching if $G[V_{M}]$ has exactly one perfect matching. Note by definition that an acyclic matching (and thus also an induced matching) is always a uniquely restricted matching, but the converse need not be true. Given a graph $G$ and a positive integer $\ell$, \textsc{Uniquely Restricted Matching} asks whether $G$ has a uniquely restricted matching of size at least $\ell$. The concept of uniquely restricted matching, motivated by a problem in Linear Algebra, was introduced by Golumbic et al. \cite{golumbic}. Some more results related to uniquely restricted matchings can be found in \cite{basteurm1, basteurm, chaudhary, francis, goddard, golumbic}.

 \subsection{Our Contributions and Methods}\label{contri}
 In Section \ref{fpt:sec}, we show that it is unlikely that there exists any $\mathsf{FPT}$-approximation algorithm for \textsc{Acyclic Matching} that approximates it to a constant factor. Our simple reduction also asserts the $\mathsf{FPT}$-inapproximability of \textsc{Induced Matching} and \textsc{Uniquely Restricted Matching}. In particular, we have the following theorem.
  \begin{restatable}{theorem}{UNION}\label{union}
Assuming $\mathsf{W[1]} \nsubseteq \mathsf{FPT}$, there
is no $\mathsf{FPT}$ algorithm that approximates any of the following to any constant when the parameter is the size of the matching:
\begin{enumerate}
    \item \textsc{Acyclic Matching}.
    \item \textsc{Induced matching}.
   \item \textsc{Uniquely Restricted Matching}. 
\end{enumerate} 
\end{restatable}
 \smallskip
 
 If $\mathcal{R}\in \{$acyclic, induced, uniquely restricted$\}$, then the $\mathcal{R}$ \emph{matching number} of $G$ is the maximum cardinality of an $\mathcal{R}$ matching among all $\mathcal{R}$ matchings in $G$. We denote by $\mathsf{AM(G)}$, $\mathsf{IM(G)}$, and $\mathsf{URM(G)}$, the \emph{acyclic matching number}, the \emph{induced matching number}, and the \emph{uniquely restricted matching number}, respectively, of $G$. Furthermore, we denote by $\mathsf{IS}(G)$ the \emph{independence number} (defined in Section \ref{GT}) of $G$.
 
 In order to prove Theorem \ref{union}, we exploit the relationship of $\mathsf{IS(G)}$ with the following: $\mathsf{AM(G)}$, $\mathsf{IM(G)}$, and $\mathsf{URM(G)}$. While the relationship of $\mathsf{AM(G)}$ and $\mathsf{IM(G)}$ with $\mathsf{IS(G)}$ is straightforward, extra efforts are needed to state a similar relation between $\mathsf{URM(G)}$ and $\mathsf{IS(G)}$, which may be of independent interest as well. Also, we note that \textsc{Induced Matching} is already known to be $\mathsf{W}[1]$-$\mathsf{hard}$, even for bipartite graphs \cite{moser}. However, for \textsc{Uniquely Restricted Matching}, Theorem \ref{union} is the first to establish the $\mathsf{W[1]}$-hardness of the problem. 

 \smallskip
 
In Section \ref{FPT:AMBTG}, we consider below-guarantee parameters for \textsc{Acyclic Matching}, i.e., parameterizations of the form $\mathsf{UB}-\ell$ for an upper bound $\mathsf{UB}$ on the size of any acyclic matching in $G$. Since the inception of  below-guarantee (and above-guarantee) parameters, there has been significant progress in the area concerning graph problems parameterized by such parameters (see a recent survey paper by Gutin and Mnich \cite{gutin}). 
It is easy to observe that in a graph $G$, the size of any acyclic matching is at most $\frac{n}{2}$, where $n=|V(G)|$. This observation yields the definition of \textsc{Acyclic Matching Below Triviality}:
\bigskip

\noindent\fbox{ \parbox{145mm}{
		\noindent\underline{\textsc{Acyclic Matching Below Triviality ($\mathsf{AMBT}$)}}:
		
		\smallskip
		\noindent\textbf{Instance:} A graph $G$ with $|V(G)|=n$ and a positive integer $\ell$.
		
		\noindent\textbf{Question:} Does $G$ have an acyclic matching of size at least $\ell$?
		
		\noindent\textbf{Parameter:} $k=n-2\ell$.}}

\bigskip

Note that for ease of exposition, we are using the parameter $n-2\ell$ instead of $\frac{n}{2}-\ell$. Next, we have the following theorem.
\begin{restatable}{theorem}{AMBTG}\label{ambt1}
There exists a randomized algorithm that solves $\mathsf{AMBT}$ with success probability at least $1-\frac{1}{e}$ in $10^{k}\cdot n^{\mathcal{O}(1)}$ time.
\end{restatable}
\smallskip

The proof of Theorem \ref{ambt1} is the most technical part of our paper. The initial intuition in proving Theorem \ref{ambt1} was to take a cue from the existing literature on similar problems (which are quite a few) like \textsc{Induced Matching Below Triviality}. However, induced matchings have many nice properties, which do not hold for acyclic matchings in general, and thus make it more difficult to characterize edges or vertices that should necessarily belong to an optimal solution of \textsc{Acyclic Matching Below Triviality}. However, there is one nice property about \textsc{Acyclic Matching}, and it is that it is closely related to the \textsc{Feedback Vertex Set} problem (defined in Section \ref{GT}) as follows. Given an instance $(G,\ell)$ of \textsc{Acyclic Matching}, where $n=|V(G)|$, $G$ has an acyclic matching of size $\ell$ if and only if there exists a (not necessarily minimal) feedback vertex set, say, $X$, of size $n-2\ell$ such that $G-X$ has a perfect matching. We use randomization techniques to find a (specific) feedback vertex set of the input graph and then check whether the remaining graph (which is a forest) has a matching of the desired size or not. Note that the classical randomized algorithm for \textsc{Feedback Vertex Set} (which is a classroom problem now) \cite{cygan} cannot be applied ``as it is'' here. In other words, since our ultimate goal is to find a matching, we cannot get rid of the vertices or edges of the input graph by applying some reduction rules ``as they are''. Instead, we can do something meaningful if we store the information about everything we delete or modify, and maintain some specific property (called Property $\mathsf{R}$ by us) in our graph. We also use our own lemma (Lemma \ref{deg2lemma}) to pick a vertex in our desired feedback vertex set with high probability. Then, using an algorithm for \textsc{Max Weight Matching} (see Section \ref{GT}), we compute an acyclic matching of size at least $\ell$, if such a matching exists. 

In light of Theorem \ref{ambt1}, for $\mathsf{AMBT}$, we ask if \textsc{Acyclic Matching} is $\mathsf{FPT}$ for natural parameters smaller than $\frac{n}{2}-\ell$. Here, an obvious upper bound on $\mathsf{AM(G)}$ is the matching number of $G$. Thus, we consider the below-guarantee parameterization $\mathsf{MM(G)}-\ell$, where $\mathsf{MM(G)}$ denotes the matching number of $G$, which yields the following problem:

\bigskip

\noindent\fbox{ \parbox{145mm}{
		\noindent\underline{\textsc{Acyclic Matching Below Maximum Matching ($\mathsf{AMBMM}$)}}:
		
		\smallskip
		\noindent\textbf{Instance:} A graph $G$ and a positive integer $\ell$.
		
		\noindent\textbf{Question:} Does $G$ have an acyclic matching of size at least $\ell$?
		
		\noindent\textbf{Parameter:} $k=\mathsf{MM(G)}-\ell$.}}
\bigskip

In \cite{furst2019some}, F\"urst and Rautenbach showed that deciding whether a given bipartite graph of
maximum degree at most 4 has a maximum matching that is also an acyclic matching is $\mathsf{NP}$-$\mathsf{hard}$. Therefore, for $k=0$, the $\mathsf{AMBMM}$ problem is $\mathsf{NP}$-$\mathsf{hard}$, and we have the following result.

\begin{restatable}{corollary}{AMBMM}\label{ambmm}
$\mathsf{AMBMM}$ is $\mathsf{para}$-$\mathsf{NP}$-$\mathsf{hard}$ even for bipartite graphs of
maximum degree at most 4.
\end{restatable}

Next, consider the following lemma, proved in Section \ref{4.1}.
\begin{restatable}{lem}{INTROLEMMA} \label{introlemma}
If $G$ has an acyclic matching $M$ of size $\ell$, then $G$ has an independent set of size at least $\ell.$ Moreover, given $M$, the independent set is computable in polynomial time.
\end{restatable}

 By Lemma \ref{introlemma}, for any graph $G$, $\mathsf{IS}(G) \geq \mathsf{AM}(G)$, which yields the following problem:
\bigskip

\noindent\fbox{ \parbox{145mm}{
		\noindent\underline{\textsc{Acyclic Matching Below Independent Set ($\mathsf{AMBIS}$)}}:
		
		\smallskip
		\noindent\textbf{Instance:} A graph $G$ and a positive integer $\ell$.
		
		\noindent\textbf{Question:} Does $G$ have an acyclic matching of size at least $\ell$?
		
		\noindent\textbf{Parameter:} $k=\mathsf{IS(G)}-\ell$.}}

\bigskip 

In \cite{panda1}, Panda and Chaudhary showed that \textsc{Acyclic Matching} is hard to approximate within factor $n^{1-\epsilon}$ for any $\epsilon>0$ unless $\mathsf{P}=\mathsf{NP}$ by giving a polynomial-time reduction from \textsc{Independent Set}. 
We notice that with a more careful analysis of the proof, the reduction given in \cite{panda1} can be used to show that $\mathsf{AMBIS}$ is $\mathsf{NP}$-$\mathsf{hard}$ for $k=0$. Therefore, we have the following result.

\begin{restatable}{corollary}{AMBIS}\label{ambis}
$\mathsf{AMBIS}$ is $\mathsf{para}$-$\mathsf{NP}$-$\mathsf{hard}$.
\end{restatable}

 We note that Hajebi and Javadi \cite{hajebi} showed that \textsc{Acyclic Matching} parameterized by treewidth ($\mathsf{tw}$) is $\mathsf{FPT}$ by using Courcelle's theorem. Since $\mathsf{tw}(G)\leq \mathsf{vc}(G)$\footnote{We denote by $\mathsf{vc}(G)$, the \emph{vertex cover number} of a graph $G$.}, this result immediately implies that \textsc{Acyclic Matching} is $\mathsf{FPT}$ with respect to the parameter $\mathsf{vc}$. We complement this result (in Section \ref{secnopoly}) by showing that it is unlikely for \textsc{Acyclic Matching} to admit a polynomial kernel when parameterized not only by $\mathsf{vc}$, but also when parameterized jointly by $\mathsf{vc}$ plus the size of the acyclic matching. In particular, we have the following theorem.

\begin{restatable}{theorem}{NOPOLYVC}\label{nopolyvc}
\textsc{Acyclic Matching} does not admit a polynomial kernel when parameterized by vertex cover number plus the size of the matching unless $\mathsf{NP}\subseteq \mathsf{coNP} \slash \mathsf{poly}$.
\end{restatable}
Parameterization by the size of a \emph{modulator} (a set of vertices
in a graph whose deletion results in a graph that belongs to a well-known and easy-to-handle graph class) is another natural
choice of investigation. We observe that with only a minor modification in our construction (in the proof of Theorem \ref{nopolyvc}), we derive the following result. 
\begin{restatable}{theorem}{NOPOLYDTC}\label{nopolydtc}
\textsc{Acyclic Matching} does not admit a polynomial kernel when parameterized by the vertex deletion distance to clique plus the size of the matching unless $\mathsf{NP}\subseteq \mathsf{coNP} \slash \mathsf{poly}$.
\end{restatable}

\section{Preliminaries}\label{prelim}
For a positive integer $k$, let $[k]$ denote the set $\{1,2,\ldots,k\}$. For a graph $G$, let $n=|V(G)|$ and $m=|E(G)|$.

\subsection{Graph-theoretic Notations and Definitions}\label{GT}

All graphs considered in this paper are simple and undirected unless stated otherwise. Standard graph-theoretic terms not explicitly defined here can be found in \cite{diestel}. For a graph $G$, let $V(G)$ denote its vertex set and $E(G)$ denote its edge set. For a graph $G$, the subgraph of $G$ induced by $S\subseteq V(G)$ is denoted by $G[S]$, where $G[S]=(S,E_{S})$ and $E_{S}=\{xy\in E(G): x,y \in S\}.$ Given a matching $M$, a vertex
$v\in V(G)$ is \textit{$M$-saturated} if $v$ is incident on an edge of $M$, that is, $v$ is an end
vertex of some edge of $M$. Given a graph $G$ and a matching $M$, we use
the notation $V_{M}$ to denote the set of $M$-saturated vertices and $G[V_{M}]$ to denote the
subgraph induced by $V_{M}$. The
\textit{matching number} of $G$ is the maximum cardinality of a matching among all matchings in $G$, and we denote it by $\mathsf{MM}(G)$. The edges in a matching $M$ are \textit{matched edges}. A matching that saturates all
the vertices of a graph is a \textit{perfect matching}. If $uv\in M$, then
$v$ is the \textit{$M$-mate} of $u$ and vice versa. Given an edge-weighted graph $G$ along with a weight function $\mathsf{w}:E(G) \rightarrow \mathbb{R}$ and a weight $\mathsf{W}\in \mathbb{R}$, \textsc{Max Weight Matching} asks whether $G$ has a matching with weight at least $\mathsf{W}$ in $G$. 

\begin{proposition}[\cite{duan}] \label{maxwt}
For a graph $G$, \textsc{Max Weight Matching} can be solved in time $\mathcal{O}(m\sqrt{n}\log(N))$, where $m=|E(G)|$, $n=|V(G)|$, and the weights are integers within the range of $[0,N]$.
\end{proposition}

The \emph{open neighborhood} of a vertex $v$ in $G$ is $N_{G}(v)=\{u\in V(G): uv\in E(G)\}$. The \emph{degree} of a vertex $v$ is $|N_{G}(v)|$, and is denoted by $d_{G}(v)$. When there is no ambiguity, we do not use the subscript $G$. The minimum and maximum degrees of graph $G$ will be denoted by $\delta(G)$ and $\Delta(G)$, respectively. A vertex $v$ with $d(v)=1$ is a \emph{pendant vertex}, and the edge incident on a pendant vertex is a \emph{pendant edge}. The \emph{distance} between two vertices in a graph $G$ is defined as the number of edges in the shortest path between them. We use the notation $\widehat{d}(u,v)$ to represent the distance between two vertices $u$ and $v$ in a graph $G$ (when $G$ is clear from the context).

An \emph{independent set} of a graph $G$ is a subset of $V(G)$ such that no two vertices in the subset have an edge between them in $G$. The
\textit{independence number} of $G$ is the maximum cardinality of an independent set among all independent sets in $G$, and we denote it by $\mathsf{IS}(G)$. Given a graph $G$ and a positive integer $\ell$, \textsc{Independent Set} asks whether there exists an independent set in $G$ of size at least $\ell$. A \emph{clique} is a subset of $V(G)$ such that every two distinct vertices in the subset are adjacent in $G$. The \emph{clique number} of $G$ is the maximum cardinality of a clique among all cliques in $G$, and we denote it by $\mathsf{\omega(G)}$. Given a graph $G$ and a positive integer $\ell$, \textsc{Clique} asks whether $\omega(G)\geq \ell$. A \emph{feedback vertex set} $\mathsf{(FVS)}$ of a graph $G$ is a subset of $V(G)$ whose removal makes $G$ a forest. Given a (multi)graph $G$ and a positive integer $\ell$, \textsc{Feedback Vertex Set} asks whether there exists a feedback vertex set in $G$ of size at most $\ell$.

A \emph{connected component} of a graph $G$ is defined as a connected subgraph of $G$ that is not part of any larger connected subgraph. A \emph{bridge} is an edge of a graph $G$ whose deletion increases the number of connected components in $G$. A \emph{factor} of a graph $G$ is a spanning subgraph of $G$ (a subgraph with vertex set $V(G)$). A \emph{$k$-factor} of a graph is a $k$-regular subgraph of order $n$. In particular, a \emph{$1$-factor} is a perfect matching. A \emph{vertex cover} of a graph $G$ is a subset $S\subseteq V(G)$ such that every edge of $G$ has at least one of its endpoints in $S$, and the size of the smallest vertex cover among all vertex covers of $G$ is known as the \emph{vertex cover number} of $G$. Let $K_{n}$ and $K_{m,n}$ denote a \emph{complete graph} with $n$ vertices and a \emph{complete bipartite graph} with parts of sizes $m$ and $n$. A \emph{clique modulator} in a graph is a vertex subset whose deletion reduces the graph to a clique. The size of a clique modulator of minimum size is known as the \emph{vertex deletion distance to a clique}. 
A \emph{binary tree} is a rooted tree where each node has at most two children. For a graph $G$ and a set $X\subseteq V(G)$, we use $G-X$ to denote $G[V(G)\setminus X]$, that is, the graph obtained from $G$ by deleting $X$.

\subsection{Parameterized Complexity}
\label{PC}
 Standard notions in Parameterized Complexity not explicitly defined here can be found in \cite{cygan, downey}.  In the framework of Parameterized Complexity, each instance of a problem $\mathrm{\Pi}$ is associated with a non-negative integer \textit{parameter} $k$. A parameterized problem $\mathrm{\Pi}$ is \textit{fixed-parameter tractable} ($\mathsf{FPT}$) if there is an algorithm that, given an instance $(I,k)$ of $\mathrm{\Pi}$, solves it in time $f(k)\cdot |I|^{\mathcal{O}(1)}$, for some computable function $f(\cdot)$. Central to Parameterized
Complexity is the following hierarchy of complexity classes:
\begin{equation} \label{eq1}
\mathsf{FPT} \subseteq \mathsf{W[1]}  \subseteq \mathsf{W[2]} \subseteq \ldots \subseteq \mathsf{XP}.
\end{equation}

All inclusions in (\ref{eq1}) are believed to be strict. In particular, \FPT $\neq$ $\mathsf{W[1]}$ under the Exponential Time Hypothesis.

\begin{definition} [Equivalent Instances] \label{equivalent}
Let $\mathrm{\Pi}$ and $\mathrm{\Pi}'$ be two parameterized problems. Two instances $(I, k)\in \mathrm{\Pi}$ and $(I', k')\in \mathrm{\Pi}'$ are \emph{equivalent} if: $(I, k)$ is a Yes-instance of $\mathrm{\Pi}$ if and only if $(I', k')$ is a Yes-instance of $\mathrm{\Pi}'$.
\end{definition}

A parameterized (decision) problem $\mathrm{\Pi}$ admits a \emph{kernel} of size $f(k)$ for some function $f$ that depends only on $k$ if the following is true: There exists an algorithm (called a \emph{kernelization algorithm}) that runs in $(|I|+k)^{\mathcal{O}(1)}$ time and translates any input instance $(I,k)$ of $\mathrm{\Pi}$ into an equivalent instance $(I',k)$ of $\mathrm{\Pi}$ such that the size of $(I',k)$ is bounded by $f(k)$. If the function $f$ is polynomial in $k$, then the problem is said to admit a \emph{polynomial kernel}. It is well-known that a decidable parameterized problem is \FPT if and only if it admits a kernel \cite{cygan}. To show that a parameterized problem does not admit a polynomial kernel on up to reasonable complexity assumptions can be done by making use of the cross-composition technique introduced by Bodlaender, Jansen, and Kratsch \cite{bodlaender}. To apply the cross-composition technique, we need the following additional definitions. Since we are using an OR-cross-composition, we give definitions tailored for this case.

\begin{definition} [Polynomial Equivalence Relation, \cite{bodlaender}]
An equivalence relation $\mathcal{R}$ on $\mathrm{\Sigma^{*}}$ is a
\emph{polynomial equivalence relation} if the following two conditions hold:
\begin{enumerate}
    \item There is an algorithm that given two strings $x,y \in \mathrm{\Sigma^{*}}$ decides whether $x$ and $y$ belong to the same equivalence class in $(|x|+|y|)^{\mathcal{O}(1)}$ time.
\item For any finite set $S\subseteq \mathrm{\Sigma^{*}}$ the equivalence relation $\mathcal{R}$ partitions the elements of $S$ into at
most $(\max_{x\in S} |x|)^{\mathcal{O}(1)}$ classes.
\end{enumerate}
\end{definition}

\begin{definition} [OR-Cross-composition, \cite{bodlaender}] 
Let $L\subseteq \mathrm{\Sigma^{*}}$ be a problem and let $Q\subseteq \mathrm{\Sigma^{*}} \times \mathbb{N}$ be a parameterized problem. We say that $L$ \emph{OR-cross-composes} into $Q$ if there is a polynomial equivalence
relation $\mathcal{R}$ and an algorithm which, given $t$ strings $x_1, x_2, \ldots, x_t$ belonging to the same equivalence class of $\mathcal{R}$, computes an instance $(x^{*},k^{*})\in \mathrm{\Sigma^{*}}\times \mathbb{N}$ in time polynomial in $\Sigma_{i=1}^{t}|x_{i}|$
such that:
\begin{enumerate}
\item $(x^{*},k^{*})\in Q \iff$ $x_{i}\in L$ for some $i\in [t]$,
\item $k^{*}$ is bounded by a polynomial in $\max_{i=1}^{t} |x_{i}|+\log t$.
\end{enumerate}
\end{definition}

\begin{proposition} [\cite{bodlaender}] \label{orcross}
If some problem $L$ is $\mathsf{NP}$-$\mathsf{hard}$ under Karp reduction and there exists an $\textsf{OR-cross-composition}$ from $L$ into
some parameterized problem $Q$, then there is no polynomial kernel for $Q$ unless $\mathsf{NP}\subseteq \mathsf{coNP} \slash \mathsf{poly}$.
\end{proposition}

\begin{definition} [$\mathsf{FPT}$-Approximation Algorithm]
Let $\mathrm{\Pi}$ be a parameterized maximization problem. An algorithm $\mathcal{A}$ is an \emph{$\mathsf{FPT}$-approximation algorithm} for $\mathrm{\Pi}$ with approximation ratio $\alpha$ if:
\begin{enumerate}
    \item For every Yes-instance $(G, \ell)$ of $\mathrm{\Pi}$ with $\mathsf{opt}(G)\geq \ell$, algorithm $\mathcal{A}$ computes a solution $S$ of $G$ such that
$|S|\cdot {\alpha}\geq \ell$. For inputs not satisfying $\mathsf{opt}(G)\geq \ell$, the output can be arbitrary.
\item If the running time of $\mathcal{A}$ is bounded by $f(\ell)\cdot |G|^{\mathcal{O}(1)}$, where $f:\mathbb{N}\rightarrow \mathbb{N}$ is computable.  \end{enumerate}
\end{definition}

\section{FPT-inapproximation Results} \label{fpt:sec}
In this section, we prove that there is no $\mathsf{FPT}$ algorithm that can approximate \textsc{Acyclic Matching}, \textsc{Induced Matching}, and \textsc{Uniquely Restricted Matching} with any constant ratio unless $\mathsf{FPT}=\mathsf{W[1]}$. To obtain our result, we need the following proposition.

\begin{proposition} [\cite{lin}] \label{lintheorem}
Assuming $\mathsf{W[1]} \nsubseteq \mathsf{FPT}$, there
is no $\mathsf{FPT}$-algorithm that approximates \textsc{Clique} to any constant.
\end{proposition}

From Proposition \ref{lintheorem}, we derive the following corollary.
\begin{corollary}\label{is}
Assuming $\mathsf{W[1]} \nsubseteq \mathsf{FPT}$, there
is no $\mathsf{FPT}$-algorithm that approximates \textsc{Independent Set} to any constant.
\end{corollary}

Before presenting our reduction from \textsc{Independent Set}, we establish (in Section \ref{4.1}) the relationship of $\mathsf{IS(G)}$ with the following: $\mathsf{AM(G)}$, $\mathsf{IM(G)}$, and $\mathsf{URM(G)}$, which will be critical for the arguments in the proof of our main theorem in this section. 
\subsection{Relation of \textsf{IS(G)} with \textsf{AM(G)}, \textsf{IM(G)} and \textsf{URM(G)}} \label{4.1}
\INTROLEMMA*
\begin{proof}
If $M$ is an acyclic matching of size $\ell$, then $G[V_{M}]$ is a forest on $2\ell$ vertices with a perfect matching. Note that $G[V_{M}]$ is a bipartite graph, and in a bipartite graph that has a perfect matching, the size of both partitions in any bipartition must be equal to each other, and hence, in our case, equal to $\ell$. This implies that $G[V_{M}]$ has an independent set of size $\ell$. Furthermore, as $G[V_{M}]$ is an induced subgraph of $G$, it is clear that $G$ has an independent set of size $\ell$ as well. \qed
\end{proof}

Since every induced matching is also an acyclic matching, the next lemma follows directly from Lemma \ref{introlemma}.
\begin{lem}\label{imlemma}
If $G$ has an induced matching $M$ of size $\ell$, then $G$ has an independent set of size at least $\ell$. Moreover, given $M$, the independent set is computable in polynomial time.
\end{lem}

Proving a similar lemma for uniquely restricted matching is more complicated. To this end, we need the following notation. Given a graph $G$, an \emph{even cycle} (i.e., a cycle with an even number of edges) in $G$ is said to be an \emph{alternating cycle} with respect to a matching $M$ if every second edge of the cycle belongs to $M$. The following proposition characterizes uniquely restricted matchings in terms of alternating cycles.
\begin{proposition} [\cite{golumbic}] \label{defurm}
Let $G$ be a graph. A matching $M$ in $G$ is uniquely restricted if and only if there is no alternating cycle with respect to $M$ in $G$.
\end{proposition}

Additionally, our proof will identify some bridges based on the following proposition.

\begin{proposition} [\cite{gabow}] \label{gabowlemma}
A graph with a unique $1$-factor has a bridge that is matched.
\end{proposition}

\begin{lem} \label{urm:lemma}
If $G$ has a uniquely restricted matching $M$ of size $\ell$, then $G$ has an independent set of size at least $\frac{\ell+1}{2}$. Moreover, given $M$, the independent set is computable in polynomial time.
\end{lem}
\begin{proof}
Let $M$ be a uniquely restricted matching in $G$ of size $\ell$. By the definition of a uniquely restricted matching, $G[V_{M}]$ has a unique perfect matching (1-factor). Let $H$ and $I$ be two sets. Initialize $H=G[V_{M}]$ and $I=\emptyset$. Next, we design an iterative algorithm, say, Algorithm \textsc{Find}, to compute an independent set in $H$ with the help of Proposition \ref{gabowlemma}. Algorithm \textsc{Find} does the following:
\medskip

While $H$ has a connected component of size at least four, go to 1.
\begin{enumerate}
    \item Pick a bridge, say, $e$, in $H$ that belongs to $M$, and go to 2. (The existence of $e$ follows from Proposition \ref{gabowlemma}.)
    \item If $e$ is a pendant edge, then remove $e$ along with its endpoints from $H$, and store the pendant vertex incident on $e$ to $I$. Else, go to 3.
    \item Remove $e$ along with its endpoints from $H$.
\end{enumerate}

After recursively applying 1-3, Algorithm \textsc{Find} arbitrarily picks exactly one vertex from each of the remaining connected components and adds them to $I$. 

Now, it remains to show that $I$ is an independent set of $H$ (and therefore also of $G$) of size at least $\frac{\ell+1}{2}$. For this purpose, we note that Algorithm \textsc{Find} gives rise to a recursive formula (defined below).

Let $R_{h}$ denote a lower bound on the maximum size of an independent set in $H$ of a maximum matching of size $h$. First, observe that if $H$ has a matching of size one, then it is clear that $H$ has an independent set of size at least 1 (we can pick one of the endpoints of the matched edge). Thus, $R_{1}=1$. Now, we define how to compute $R_{h}$ recursively.
\begin{equation} \label{urm} R_{h}=\min \{R_{h-1}+1, \min_{\displaystyle{\small{1\leq i \leq h-2}}} R_{i}+R_{h-i-1}\}. \end{equation}

The first term in (\ref{urm}) corresponds to the case where the matched bridge is a pendant edge. On the other hand, the second term in (\ref{urm}) corresponds to the case where the matched bridge is not a pendant edge. In this case, all the connected components have at least one of the matched edges. Next, we claim the following, 
\begin{equation}\label{urm1}
R_{h}\geq \frac{h+1}{2}.
\end{equation} 

We prove our claim by applying induction on the maximum size of a matching in $H$. Recall that $R_{1}=1$. Next, by the induction hypothesis, assume that (\ref{urm1}) is true for all $k<h$. Note that since $h-1<h$, (\ref{urm1}) is true for $k=h-1$, i.e., $R_{h-1}\geq \frac{h}{2}$. To prove that (\ref{urm1}) is true for $k=h$, we first assume that the first term, i.e., $R_{h-1}+1$ gives the minimum in (\ref{urm}). In this case, $R_{h}\geq \frac{h}{2}+1=\frac{h+2}{2}\geq \frac{h+1}{2}.$ Next, assume that the second term gives the minimum in (\ref{urm}) for some $i', 1\leq i' \leq h-2$. In this case, note that $i'< h$ and $h-i'-1\leq h-2<h$, and thus, by the induction hypothesis, (\ref{urm1}) holds for both $R_{i'}$ and $R_{h-i'-1}$. Therefore, $R_{h}\geq \frac{i'+1}{2}+\frac{h-i'}{2}=\frac{h+1}{2}$. \qed
\end{proof}

\begin{remark}
Throughout this paper, let $\textsc{Restricted}\in \{\textsc{Acyclic}, \textsc{Induced}, \textsc{Uniquely}$
$\textsc{Restricted}\}$ and $\mathcal{R}\in \{$acyclic, induced, uniquely restricted$\}$.
\end{remark}
\begin{figure}[t]
 \centering
    \includegraphics[scale=0.9]{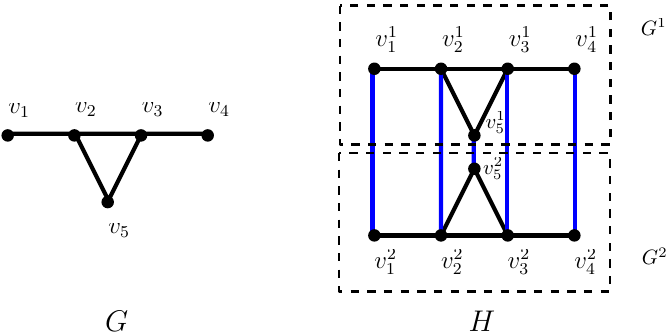}
    \caption{The construction of $H$ from $G$. The edges shown in blue color represent vertical edges.}
    \label{fig3}
\end{figure}

By Lemmas \ref{introlemma}-\ref{urm:lemma}, we have the following corollary.
\begin{corollary}\label{newlemma}
If a graph $G$ has an $\mathcal{R}$ matching $M$ of size $\ell$, then $G$ has an independent set of size at least $\frac{\ell+1}{2}$. Moreover, given $M$, the independent set is computable in polynomial time.
\end{corollary}

\subsection{Construction} \label{union:sec}
Given a graph $G$, where $V(G)=\{v_{1},\ldots,v_{n}\}$, we construct a graph $H=\mathsf{reduce}(G)$ as follows. Let $G^{1}$ and $G^{2}$ be two copies of $G$. Let $V^{1}=\{v_{1}^{1},\ldots,v_{n}^{1}\}$ and $V^{2}=\{v_{1}^{2},\ldots,v_{n}^{2}\}$ denote the vertex sets of $G^{1}$ and $G^{2}$, respectively. Let $V(H)=V(G^{1})\cup V(G^{2})$ and $E(H)=E(G^{1})\cup E(G^{2})\cup \{v_{i}^{1}v_{i}^{2}: i\in [n]\}$. See Figure \ref{fig3} for an illustration of the construction of $H$ from $G$. Let us call the edges between $G^{1}$ and $G^{2}$ \emph{vertical edges}.
\begin{lem} \label{rmatch}
Let $G$ and $H$ be as defined above. If $H$ has an $\mathcal{R}$ matching of the form $M=\{v_{1}^{1}v_{1}^{2}, \ldots,v_{p}^{1}v_{p}^{2}\}$, then $I_{M}=\{v_{1},  \ldots,v_{p}\}$
is an independent set of $G$.
\end{lem}

\begin{proof}
If $I_{M}$ is an independent set of $G$, then we are done. So, targeting a contradiction, assume that there exist distinct $v_{i},v_{j}\in I_{M}$ such that $v_{i}v_{j}\in E(G)$. By the definition of $H$, $v_{i}^{1}v_{j}^{1}, v_{i}^{2}v_{j}^{2} \in E(H)$. This implies that $v_{i}^{1},v_{i}^{2},v_{j}^{2},v_{j}^{1},v_{i}^{1}$ forms a cycle in $H$, and thus $M$ is neither an induced matching nor an acyclic matching in $H$. Furthermore, observe that $v_{i}^{1},v_{i}^{2},v_{j}^{2},v_{j}^{1},v_{i}^{1}$ is an alternating cycle of length four in $H$, therefore, by Proposition \ref{defurm}, $M$ is not a uniquely restricted matching in $H$. Thus, we necessarily reach a contradiction. \qed
\end{proof}

\subsection{Hardness of Approximation Proof} 

To prove the main theorem of this section (Theorem \ref{union}), we first suppose that \textsc{Restricted Matching} can be approximated within a ratio of $\alpha>1$, where $\alpha\in \mathbb{R}^{+}$ is a constant, by some FPT-approximation algorithm, say, Algorithm $\mathcal{A}$. 

By the definition of $\mathcal{A}$, the following is true:
 \begin{enumerate}
     \item [i)] If $H$ does not have an $\mathcal{R}$ matching of size $\ell$, then the output of $\mathcal{A}$ is arbitrary (indicating that $(H,\ell)$ is a No-instance).
     \item [ii)] If $H$ has an $\mathcal{R}$ matching of size $\ell$, then $\mathcal{A}$ returns an $\mathcal{R}$ matching, say, $X$, such that $|X|\geq \frac{\ell}{\alpha}$ and $|X|\leq \mathsf{opt(H)}$, where $\mathsf{opt(H)}$ denotes the optimal size of an $\mathcal{R}$ matching in $H$.
 \end{enumerate} 
 
  Next, we propose an FPT-approximation algorithm, say, Algorithm $\mathcal{B}$, to compute an $\mathsf{FPT}$-approximate solution for \textsc{Independent Set} as follows. Given an instance $(G,\ell)$ of \textsc{Independent Set}, Algorithm $\mathcal{B}$ first constructs an instance $(H,\ell)$ of \textsc{Restricted Matching}, where $H=\mathsf{reduce}(G)$ (see Section \ref{union:sec}). Algorithm $\mathcal{B}$ then solves $(H,\ell)$ by using Algorithm $\mathcal{A}$. If $\mathcal{A}$ returns an $\mathcal{R}$ matching $X$, then $\mathcal{B}$ returns an independent set of size at least $\frac{|X|}{8} (\geq \frac{\ell}{8\alpha})$. Else, the output is arbitrary.
 \medskip
 
 Now, it remains to show that $\mathcal{B}$ is an $\mathsf{FPT}$-approximation algorithm for \textsc{Independent Set} with an approximation factor of $\beta>1$, where $\beta\in \mathbb{R}^{+}$, which we will show with the help of the following two lemmas. 

 \begin{lem} \label{inap1}
 Algorithm $\mathcal{B}$ approximates \textsc{Independent Set} within a constant factor $\beta>1$, where $\beta\in \mathbb{R}^{+}$.
 \end{lem}
 \begin{proof}
 Let $\beta=8\alpha$. First, observe that if $G$ does not have an independent set of size $\ell$, then $\mathcal{B}$ can return any output. So, we next suppose that $G$ has an independent set, say, $S=\{v_{1},\ldots,v_{\ell}\}$, of size $\ell$. Then, notice that $\{v_{1}^{1}v_{1}^{2},\ldots,v_{\ell}^{1}v_{\ell}^{2}\}$ is an $\mathcal{R}$ matching of size $\ell$ in $H$. Therefore, $(H, \ell)$ must be a Yes-instance, and in this case, $\mathcal{A}$ must return a solution $X\neq \emptyset$ such that $\frac{\ell}{\alpha}\leq|X|\leq \mathsf{opt}(H)$.
 
 Now, observe that either at least $\frac{|X|}{2}$ edges are vertical edges or at least $\frac{|X|}{2}$ edges belong to $G^{1}\cup G^{2}$. If at least $\frac{|X|}{2}$ edges are vertical edges, then by Lemma \ref{rmatch}, $G$ has an independent set of size at least $\frac{|X|}{2}$. For instance, if $\{v_{1}^{1}v_{1}^{2},\ldots,v_{p}^{1}v_{p}^{2}\}$, where $p\geq \frac{|X|}{2}$, is an $\mathcal{R}$ matching of $H$, then $\{v_{1},\ldots,v_{p}\}$ is an independent set of $G$, and in this case, $\mathcal{B}$ returns $\{v_{1},\ldots,v_{p}\}$. As $|X|\geq \frac{\ell}{\alpha}$, so $p\geq \frac{\ell}{2\alpha}$, and thus $\mathcal{B}$ is an $\beta\geq 2\alpha$ approximation algorithm for the \textsc{Independent Set} problem.
 
 On the other hand, if at least $\frac{|X|}{2}$ edges belong to $G^{1}\cup G^{2}$, then note that either $G^{1}$ or $G^{2}$ has an $\mathcal{R}$ matching of size at least $\frac{|X|}{4}$. Without loss of generality, let $G^{1}$ have an $\mathcal{R}$ matching of size at least $\frac{|X|}{4}$. Since $G^{1}$ is a copy of $G$, by Corollary \ref{newlemma}, $G$ has an independent set of size at least $ \frac{|X|+4}{8}$ $\geq \frac{|X|}{8}$. As $|X|\geq \frac{\ell}{\alpha}$, so $\frac{|X|}{8}\geq \frac{\ell}{8\alpha}$, and thus $\mathcal{B}$ is an $\beta=8\alpha$ approximation algorithm for the \textsc{Independent Set} problem. \qed
 \end{proof}
 
  \begin{lem} \label{inap0}
  Algorithm $\mathcal{B}$ runs in $\mathsf{FPT}$ time.
 \end{lem}
 \begin{proof}
 By Corollary \ref{newlemma}, it is clear that given an $\mathcal{R}$ matching of size $\ell$ in an input graph $G$, one can compute an independent set of size at least $\frac{\ell+1}{2}$ in polynomial time. Since Algorithm $\mathcal{A}$ runs in $\mathsf{FPT}$ time, therefore, Algorithm $\mathcal{B}$ runs in $\mathsf{FPT}$ time with some polynomial overheads in the size of the input graph, and thus Algorithm $\mathcal{B}$ runs in $\mathsf{FPT}$ time. \qed
 \end{proof}
 
 By Corollary \ref{newlemma} and Lemmas \ref{inap1} and \ref{inap0}, we have the following theorem.
\UNION*
\section{Below-guarantee Parameters}
\subsection{FPT Algorithm for \textsf{AMBT}} \label{FPT:AMBTG}

In this section, we prove that $\mathsf{AMBT}$ is $\mathsf{FPT}$ by giving a randomized algorithm that runs in time $10^{k}\cdot n^{\mathcal{O}(1)}$, where
$n=|V(G)|$ and $k=n-2\ell$.

 First, we define some terminology that is crucial for proceeding further in this section. A graph $G$ \emph{has property $\mathsf{R}$} if $\delta(G)\geq 2$ and no two adjacent vertices of $G$ have degree exactly 2. A path $P$ is a \emph{maximal degree-2 path} in $G$ if: $(i)$ it has at least two vertices, $(ii)$ the degree of each vertex in $P$ (including the endpoints) is exactly 2, and $(iii)$ it is not contained in any other degree-2 path. If we replace a maximal degree-2 path $P$ with a single vertex, say, $v_{P}$, of degree exactly 2 (in $G$), then we call this operation \textsc{Path-Replacement($P$,$v_{P}$)} (note that the neighbors of $v_{P}$ are the neighbors of the endpoints of $P$ that do not belong to $P$). Furthermore, we call the newly introduced vertex (that replaces a maximal degree-2 path in $G$) \emph{virtual vertex}. Note that if both endpoints of $P$ have a common neighbor, then this gives rise to multiple edges in $G$. Next, if there exists a cycle, say, $C$, of length $p\geq 2$ such that the degree of each vertex in $C$ is exactly 2 (in $G$), then the \textsc{Path-Replacement} operation also identifies such cycles and replaces each of them with a virtual vertex having a self-loop, and the corresponding maximal degree-2 path, in this case, consists of all the vertices of $C$. Therefore, it is required for us to consider $\mathsf{AMBT}$ in the more general setting of multigraphs, where the graph obtained after applying the \textsc{Path-Replacement} operation may contain multiple edges and self-loops. We also note that multiple edges and self-loops are cycles.

We first present a lemma (Lemma \ref{deg2lemma}) that is crucial to prove the main result (Theorem \ref{ambt1}) of this section.

\begin{lem}\label{deg2lemma}
Let $G$ be a graph on $n$ vertices with the property $\mathsf{R}$. Then, for every feedback vertex set $X$ of $G$, more than $\frac{|E(G)|}{5}$ of the edges of $G$ have at least one endpoint in $X$.
\end{lem}
\begin{proof}
Let $X$ be an arbitrary but fixed vertex cover of $G$, and let $F=G-X$. To prove our lemma, we need to show that $|E(G)\setminus E(F)|>\frac{1}{5}|E(G)|$. Let $Y$ denote the set of edges in $G$ with one endpoint in $X$ and the other endpoint in $V(F)$. Next, let us partition the set $V(F)$ in the following sets.
\medskip

\noindent $V_{\leq1}=\{v\in V(F): d_{F}(v)\leq 1\}$.\\
$V^{*}_{2}=\{v\in V(F): d_{F}(v)=2$ and $d_{G}(v)=2$\}.\\
$V_{2}=\{v\in V(F): d_{F}(v)=2$ and $d_{G}(v)\geq3\}.$\\
$V_{\geq3}=\{v\in V(F): d_{F}(v)\geq 3\}$.

In any forest, the number of vertices that have degree exactly 1 is strictly greater than the number of vertices that have degree at least 3. Since $F$ is a forest, we have our first inequality as follows.
\begin{equation}\label{eq:0}
|V_{\leq1}|> |V_{\geq3}|.\end{equation}

\begin{figure}[t]
 \centering
    \includegraphics[scale=1]{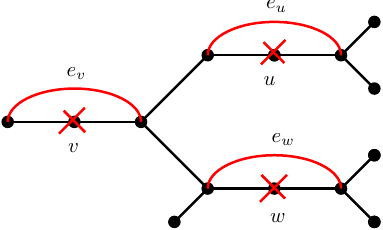}
    \caption{An illustration of the replacement of vertices in $V_{2}^{*}$ by their corresponding edges.}
    \label{fig5}
\end{figure}

Next, observe that if we replace every vertex $v\in V^{*}_{2}$ with an edge, say, $e_{v}$, then $F$ still remains a forest. See Figure \ref{fig5} for an illustration.
Since the number of edges in a forest is bounded from above by the number of vertices in it, we have the following inequality,
\begin{equation}\label{eq:1}
    |V_{2}^{*}|\leq |V_{\leq 1}|+|V_{2}|+|V_{\geq3}|.
\end{equation}
By using (\ref{eq:0}) in (\ref{eq:1}), we get
\begin{equation}\label{eq:2}
    |V_{2}^{*}|< 2|V_{\leq1}|+|V_{2}|.
\end{equation}

Since vertices in $V_{\leq 1}$ and $V_{2}$ have degree at least 2 and at least 3 in $G$, respectively, both $V_{\leq1}$ and $V_{2}$ contribute at least one edge to $Y$. As $|E(G)\setminus E(F)|\geq |Y|$, we have the following inequality, 
\begin{equation}\label{eq:4}
 |E(G)\setminus E(F)|\geq |V_{\leq1}|+|V_{2}|.    
\end{equation}

Next, note that the edges that do not have any endpoint in $X$ are exactly the edges in $F$, and the total number of such edges must be bounded by the sum of the size of the sets $V_{\leq 1},V_{2}, V^{*}_{2},$ and $V_{\geq3}$. Therefore, we have the following inequality,
\begin{equation}\label{eq:5}
|E(F)|\leq |V_{\leq1}|+|V_{2}^{*}|+|V_{2}|+|V_{\geq3}|.
\end{equation}
By using (\ref{eq:0}) and (\ref{eq:2}) in (\ref{eq:5}), we get
\begin{equation}\label{eq:6}
|E(F)|< 4|V_{\leq1}|+2|V_{2}|.
\end{equation}

Let $|V_{\leq1}|+|V_{2}|=\alpha$ and $4|V_{\leq1}|+2|V_{2}|=\beta$. Since $|V_{\leq1}|, |V_{2}|, |V^{*}_{2}|,|V_{\geq 3}|\geq 0$, it is clear that $\alpha\geq \frac{\beta}{4}$. The fraction of edges having at least one endpoint in $X$ is exactly $\frac{|E(G)\setminus E(F)|}{|E(G)\setminus E(F)|+|E(F)|}$. By (\ref{eq:4}) and (\ref{eq:6}), $|E(F)|<\beta\leq4\alpha\leq 4\cdot |E(G)\setminus E(F)|$. This implies that $|E(G)\setminus E(F)|+|E(F)|< 5\cdot |E(G)\setminus E(F)|$.
Therefore, we have $\frac{|E(G)\setminus E(F)|}{|E(G)\setminus E(F)|+|E(F)|}> \frac{1}{5}.$ Hence, the lemma is proved. \qed
\end{proof}

\medskip

Now, consider Algorithm \ref{algo22}.
\begin{algorithm}[t]

		\KwIn{A graph $G$ and a positive integer $k$;}
  
		\KwOut{A set $\widehat{X}$ of size at most $k$, a set $Z$, and a set $A$ or $\mathsf{No}$;}
		Initialize $Z\leftarrow \emptyset$, $A\leftarrow \emptyset$, $\widehat{X}\leftarrow \emptyset$; 
		
		\While{$(V(G)\neq \emptyset)$}{
		\While{$(\delta(G)\leq 1)$}{
			Pick a vertex $v\in V(G)$ such that $d(v)\leq1$;  
			
			$Z\leftarrow Z\cup \{v\}$;
			
				$V(G)\leftarrow V(G) \setminus \{v\}$; }
		\While{$($there exists a maximal degree-2 path $P$ in $G$$)$}{
		  \textsc{Path-Replacement($P$,$v_{P}$)};
		  
		  $A \leftarrow A\cup\{(P,v_{P})\};$
		  }
		  	
		\If{$( k> 0$ and $G$ has a cycle)}{
		\If{$(G$ has a self-loop at some $v)$}{$\widehat{X}\leftarrow \widehat{X}\cup \{v\}$;
			
				$V(G)\leftarrow V(G) \setminus \{v\}$; 
				
				$k\leftarrow k-1;$}
		\Else{
		Pick an edge $e\in E(G)$;
		
		Pick an endpoint $v$ of $e$;   
			
			$\widehat{X}\leftarrow \widehat{X}\cup \{v\}$;
			
				$V(G)\leftarrow V(G) \setminus \{v\}$; 
				
				$k\leftarrow k-1;$
				
			}	}
				\ElseIf{$(k\leq0$ and $G$ has a cycle$)$}
 {\Return{$\mathsf{No}$};}
		  
}
\Return{$\widehat{X},Z,A$};

		\caption{}
	\label{algo22}

 \end{algorithm}
	\medskip

Observe that the task of Algorithm \ref{algo22} is first to modify an input graph $G$ to a graph that has property $\mathsf{R}$. By abuse of notation, we call this modified graph $G$. Since $G$ is non-empty and $G$ has property $\mathsf{R}$, then $G$ definitely has a cycle, and by Lemma \ref{deg2lemma}, with probability at least $\frac{1}{10}$, we pick one vertex, say $v$, that belongs to a specific feedback vertex set of $G$. We store this vertex in a set $\widehat{X}$. After removing $v$ from $G$, we also decrease $k$ by $1$. We again repeat the process until either the graph becomes empty or $k$ becomes non-positive while there are still some cycles left in the graph; we return $\mathsf{No}$ in the latter case, and the sets $A$, $Z$, and $\widehat{X}$ in the former case. 

\begin{remark}
We call the set $\widehat{X}$ returned by Algorithm \ref{algo22} a \emph{virtual feedback vertex set}.
\end{remark}

Given an input graph $G$ and a positive integer $k$, let $\widehat{X}$ be a virtual feedback vertex set returned by Algorithm \ref{algo22}. Note that the set $\widehat{X}$ contains a combination of virtual vertices and the vertices from the set $V(G)$. Let $\widehat{V}\subseteq \widehat{X}$ be the set of virtual vertices. If $v_{P}\in \widehat{V}$, then there exists some maximal degree-2 path $P$ such that $(v_{P},P)\in A$. If all the vertices in $P$ are from $V(G)$, then we say that the set of vertices in $P$ is \emph{safe} for $v_{P}$. On the other hand, if the path $P$ contains some virtual vertices, then note that there exist maximal degree-2 paths corresponding to these virtual vertices as well. In this case, we recursively replace the virtual vertices present in $P$ with their corresponding maximal degree-2 paths until we obtain a set that contains vertices from $V(G)$ only, and we say that these vertices are safe for $v_{P}$. The process of obtaining a set of safe vertices corresponding to virtual vertices is shown in Figure \ref{1}. Note that, for the graph shown in Figure \ref{1} (i), if $\widehat{X}=\{v_{P_{3}},v_{P_{4}},v_{P_{6}}\}$ is a virtual feedback vertex set returned by Algorithm \ref{algo22} corresponding to $X=\{i,k,a\}$, then the safe set corresponding to $v_{P_{3}}$ is $\{i,j\}$, corresponding to $v_{P_{4}}$ is $\{k,l\}$, and corresponding to $v_{P_{6}}$ is $\{a,b,c,d,e,f,g,h\}$.

\begin{figure}[t]
 \centering
    \includegraphics[scale=0.7]{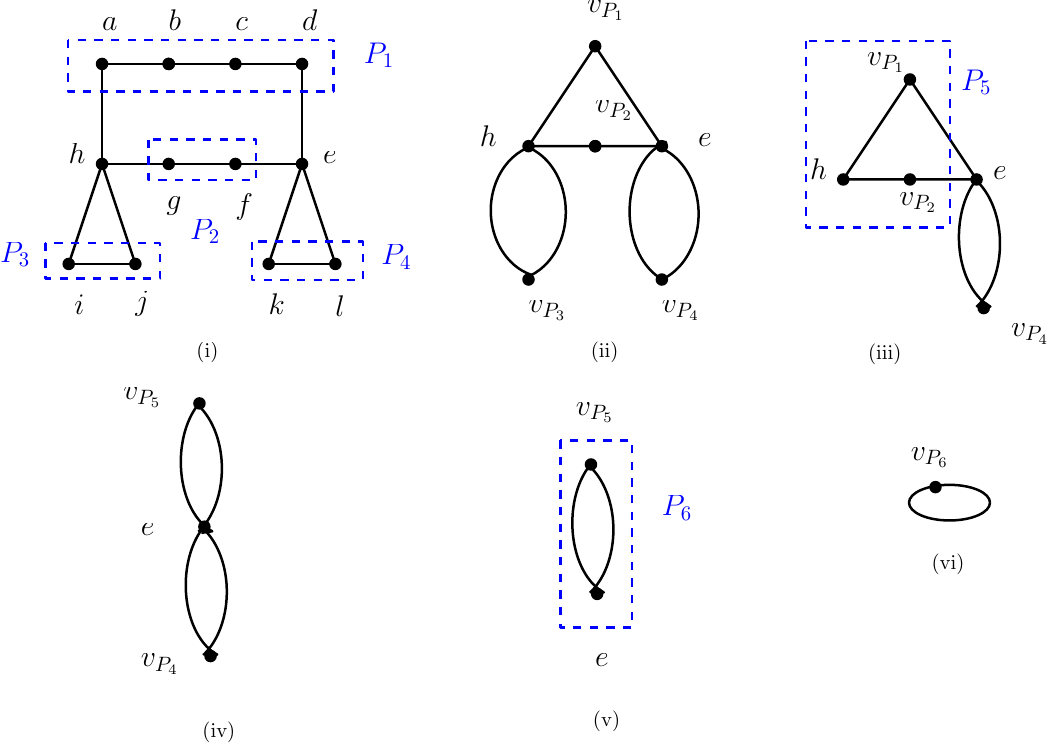}
    \caption{The graph in (i) has four maximal degree-2 paths. After applying the \textsc{Path-Replacement} operation to paths $P_1$-$P_4$, we obtain the graph shown in (ii), which has property $\mathsf{R}$. We assume that Algorithm \ref{algo22} picks $v_{P_{3}}$ in $\widehat{X}$. After removing  $v_{P_{3}}$, we obtain the graph shown in (iii). The graph in (iii) has a maximal degree-2 path $P_{5}$. After applying the  \textsc{Path-Replacement} operation to $P_5$, we obtain the graph shown in (iv). We assume that Algorithm \ref{algo22} picks $v_{P_{4}}$ in $\widehat{X}$. After removing  $v_{P_{4}}$, we obtain the graph shown in (v). Note that the \textsc{Path-Replacement} operation identifies the cycle shown in (v) as a maximal degree-2 path and replaces it with a virtual vertex ($v_{P_{6}}$) with a self-loop, as shown in (vi). Algorithm \ref{algo22} then picks $v_{P_{6}}$ in $\widehat{X}$.}
    \label{1}
\end{figure}

\begin{remark}
Throughout this section, if $\widehat{X}$ is a virtual feedback vertex set returned by Algorithm \ref{algo22}, then let $\widehat{V}\subseteq \widehat{X}$ denote the set of virtual vertices.
\end{remark}

Next, consider the following definition. 

\begin{definition} \label{defcompatible} Let $\widehat{X}$ be a virtual feedback vertex set returned by Algorithm \ref{algo22} if given as input a graph $G$ and a positive integer $k$. A set $X\subseteq V(G)$ is \emph{compatible} with $\widehat{X}$
if the following hold. 
\begin{enumerate}
    \item For every $v_{P}\in \widehat{V}$, $X$ contains at least one vertex from the set of safe vertices corresponding to $v_{P}$.
    \item For every $v\in \widehat{X}\setminus \widehat{V}$, $v$ belongs to $X$.
    \item $|X|\leq k$.
\end{enumerate} 
\end{definition}

Next, given a graph $G$, consider the following two reduction rules.
\medskip

\noindent RR1: If there is a vertex $v\in V(G)$ such that $d(v)\leq1$, then set $V(G)= V(G) \setminus \{v\}$.

\noindent RR2: If there is a maximal degree-2 path $P$ in $G$, then apply \textsc{Path-Replacement ($P$,$v_{P}$)}.

Note that after recursively applying RR1 and RR2 to a (multi)graph $G$, either we get an empty graph or a graph with property $\mathsf{R}$.

\begin{lem} \label{lm1}
Given a graph $G$ and a positive integer $k$, if $\widehat{X}$ is a virtual feedback vertex set returned by Algorithm \ref{algo22}, then any set $X\subseteq V(G)$ compatible with $\widehat{X}$ is a feedback vertex set of $G$.
\end{lem}
\begin{proof}
We will prove the lemma by induction on $|V(G)|$. For the base case, let $V(G)=\emptyset$. Since every set is trivially a feedback vertex set of $G$ (as $G=\emptyset$), the base case holds.

Now, let us assume that for any graph $G$ with $|V(G)|\leq t$, Lemma \ref{lm1} holds. Next, we claim the following.

\begin{clm}
For any graph $G$ with $|V(G)|=t+1$, if $\widehat{X}$ is a virtual feedback vertex set returned by Algorithm \ref{algo22}, then any set $X\subseteq V(G)$ compatible with $\widehat{X}$ is a feedback vertex set of $G$.
\end{clm}
\begin{proof}
In order to prove our claim, first, let $G^{\mathsf{R}}$ denote the graph obtained from $G$ by applying RR1 and RR2 exhaustively on $G$. Next, consider the following cases based on whether $G=G^{\mathsf{R}}$ or not: 
\medskip

\noindent \textbf{Case 1:} \bm{$G^{\mathsf{R}}=G.$} Let $v$ be a (first) vertex that Algorithm \ref{algo22} adds to the virtual feedback vertex set of $G$. Now, consider the graph $G-\{v\}$. Let $\widehat{Y}$ be a virtual feedback vertex set of $G-\{v\}$ computed afterwards by Algorithm \ref{algo22}. This implies that $\widehat{Y}\cup \{v\}$ is a virtual feedback vertex set of $G$ returned by Algorithm \ref{algo22}. Since $|V(G-\{v\})|= t$, by the induction hypothesis, any set $Y\subseteq V(G-\{v\})$ that is compatible with $\widehat{Y}$ is a feedback vertex set of $G-\{v\}$. 

Next, we claim that any set $X\subseteq V(G)$ that is compatible with $\widehat{Y}\cup \{v\}$ is a feedback vertex set of $G$. Note that any set that is compatible with $\widehat{Y}\cup \{v\}$ must be of the form $Y\cup \{v\}$, where $Y$ is compatible with $\widehat{Y}$ (see Definition \ref{defcompatible}). Since any set $Y\subseteq V(G-\{v\})$ that is compatible with $\widehat{Y}$ is a feedback vertex set of $G-\{v\}$, we only need to show (in order to prove our claim) that every cycle in $G-Y$ must contain $v$. Targeting a contradiction, let there exist a cycle, say, $C$, in $G-Y$ that does not contain $v$. This implies that $C$ must be a cycle in $G-(Y\cup \{v\})$ as well. This leads to a contradiction to the fact that $Y$ is a feedback vertex set of $G-\{v\}$. Thus, $X$ is a feedback vertex set of $G$. \smallskip

\noindent \textbf{Case 2:} \bm{$G^{\mathsf{R}}\neq G.$} Let $G'$ be the graph obtained from $G$ by a single application of a reduction rule (i.e., if RR1 is applicable on $G$, then we apply RR1 exactly once on $G$; otherwise, we apply RR2 exactly once on $G$). Since the application of either RR1 or RR2 reduces the number of vertices of $G$ by at least one (by the definitions of RR1 and RR2), it is clear that $|V(G')|<|V(G)|=t+1$. So, by the induction hypothesis, if $\widehat{Y}$ is a virtual feedback vertex set of $G'$ returned by Algorithm \ref{algo22}, then any set $Y'\subseteq V(G')$ that is compatible with $\widehat{Y}$ is a feedback vertex set of $G'$. Note that $\widehat{Y}$ is a virtual feedback vertex set of $G$ as well (returned by Algorithm \ref{algo22}). Now, it is left to show that any set $Y\subseteq V(G)$ that is compatible with $\widehat{Y}$ is a feedback vertex set of $G$.

First, assume that we have applied RR1 on $G$. Since the vertex $v$ in $V(G)\setminus V(G')$ has degree at most $1$, it does not belong to any cycle of $G$. In other words, the sets of cycles in $G$ and $G'$ are the same. Consider the set $Y'=Y\setminus \{v\}$ (possibly $Y'=Y$, if $v\notin Y$) in $G'$. Note that $Y'$ is compatible with $\widehat{Y}$ (as $Y$ is compatible with $\widehat{Y}$ and $v \notin \widehat{Y}$). Since $Y'$ is a feedback vertex set of $G'$, $Y$ is a feedback vertex set of $G$. 

Next, assume that we have applied RR2 on $G$. Let $P$ be the maximal degree-2 path in $G$ that has been replaced by a virtual vertex, say, $v_{P}$, by RR2. Now, let us consider the following two cases based on whether $Y\subseteq V(G')$ or not. \smallskip

\noindent \textbf{Case 2.1:} \bm{$Y\subseteq V(G').$} Note that, in this case, $v_{P}\notin \widehat{Y}$. Else, some vertex from $P$ must belong to $Y$ (see Definition \ref{defcompatible}), and this contradicts the fact that $Y\subseteq V(G')$. Next, consider the set $Y'=Y$ in $G'$. If $Y'$ is a feedback vertex set of $G$, then we are done. So, assume otherwise. This implies that there must exist a cycle, say, $C$, in $G-Y'$. Since $Y'$ is a feedback vertex set of $G'$ (and not $G$), $C$ must contain $P$ (note that if $C$ contains at least one vertex from $P$, then it must contain every vertex from $P$, as $P$ is a maximal degree-2 path). Let $C'$ be the cycle (in $G'$) obtained from $C$ by replacing $P$ with $v_{P}$. Then, after applying the \textsc{Path-Replacement$(P,v_{P})$} operation, $C'$ must be a cycle in $G'-Y'$. This leads to a contradiction to the fact that $Y'$ is a feedback vertex set of $G'$. Thus, $Y'(=Y)$ is a feedback vertex set of $G$. \smallskip

\noindent \textbf{Case 2.2:} \bm{$Y\nsubseteq V(G').$} Note that, in this case, $Y$ must contain at least one vertex from $P$. Consider the set $Y'=(Y\cap V(G'))\cup \{v_{P}\}$ in $G'$. Since $V(G)\setminus V(G')=V(P)$ and $v_{P}\in Y'$, $Y'$ is compatible with $\widehat{Y}$. If $Y$ is a feedback vertex set of $G$, then we are done. So, assume otherwise. This implies that there must exist a cycle, say, $C$, in $G-Y$. Since $Y'$ is a feedback vertex set of $G'$ and $P$ is a maximal degree-2 path, $C$ must contain $P$. This is a contradiction as $Y$ contains a vertex from $P$. Thus, $Y$ is a feedback vertex set of $G$. \qed
\end{proof}

Since the result is also true for the graph with $|V(G)|=t+1$, by the mathematical induction, the lemma holds. \qed
\end{proof}

\begin{lem} \label{lm2}
Let $G$ be a graph and $k$ be a positive integer. Then, for any feedback vertex set $X$ of $G$ of size at most $k$, with probability at least $10^{-k}$, Algorithm \ref{algo22} returns a virtual feedback vertex set $\widehat{X}$ such that $X$ is compatible with $\widehat{X}$.
 \end{lem}
\begin{proof}
We will prove the lemma by induction on $|V(G)|$. For the base case, let $V(G)=\emptyset$. In this case, any feedback vertex set of $G$ is the empty set. Furthermore, note that if $V(G)=\emptyset$, then Algorithm \ref{algo22} returns $\widehat{X}=\emptyset$ with probability $1$. Since any set is trivially compatible with the empty set, the base case holds.

Now, let us assume that for any graph $G$ with $|V(G)|\leq t$, Lemma \ref{lm2} holds. Next, we claim the following.

\begin{clm}
For any graph $G$ with $|V(G)|=t+1$, if $X$ is a (fixed but arbitrary) feedback vertex set of $G$ of size at most $k$, then, with probability at least $10^{-k}$, Algorithm \ref{algo22} returns a virtual feedback vertex set $\widehat{X}$ such that $X$ is compatible with $\widehat{X}$. 
\end{clm}
\begin{proof} Note that we can assume, without loss of generality, that $X$ is a minimal feedback vertex set of $G$, because if $X$ is compatible with $\widehat{X}$, then any set $\overline{X}$ such that $X\subseteq \overline{X}$ is also compatible with $\widehat{X}$ (see Definition \ref{defcompatible}). In order to prove our claim, first, let $G^{\mathsf{R}}$ denote the graph obtained from $G$ by applying RR1 and RR2 exhaustively on $G$. Next, consider the following cases based on whether $G=G^{\mathsf{R}}$ or not: 
\medskip

\noindent \textbf{Case 1:} \bm{$G^{\mathsf{R}}=G.$} In this case, note that if $X$ is a feedback vertex set of $G$, then $X$ is a feedback vertex set of $G^{\mathsf{R}}$ as well (as $G^{\mathsf{R}}=G$). Therefore, since $G^{\mathsf{R}}$ has the property $\mathsf{R}$, Lemma \ref{deg2lemma} implies that with a probability of at least $\frac{1}{10}$, Algorithm \ref{algo22} picks a vertex, say, $v$, in $X$. Next, note that $X\setminus \{v\}$ is a feedback vertex set of $G-\{v\}$. Since $|V(G-\{v\})|= t$ and $|X\setminus \{v\}|=k-1$, by the induction hypothesis, with probability at least $10^{-(k-1)}$, Algorithm \ref{algo22} returns a virtual feedback vertex set, say, $\widehat{Y}$, of $G-\{v\}$ such that $X\setminus \{v\}$ is compatible with $\widehat{Y}$. Then, it follows (by Definition \ref{defcompatible}) that $X$ is compatible with $\widehat{Y}\cup \{v\}$. Furthermore, note that the probability of returning $\widehat{Y}\cup \{v\}$ by Algorithm \ref{algo22} is at least $10^{-(k-1)}\cdot 10^{-1}=10^{-k}$.\\

\noindent \textbf{Case 2:} \bm{$G^{\mathsf{R}}\neq G.$} Let $G'$ be the graph obtained from $G$ by a single application of a reduction rule (i.e., if RR1 is applicable on $G$, then we apply RR1 exactly once on $G$; otherwise, we apply RR2 exactly once on $G$). Next, we consider the following two sub-cases.
\medskip

\noindent \textbf{Subcase 2.1:} \bm{$X\subseteq V(G').$}  First, we claim that $X$ is a feedback vertex set of $G'$. First, assume that we have applied RR1 on $G$. Since we remove a vertex of degree at most $1$ in RR1, the claim follows immediately. Next, assume that we have applied RR2 on $G$. Targeting a contradiction, assume that $X$ is not a feedback vertex set of $G'$. This means that after applying the \textsc{Path-Replacement$(P,v_{P})$} operation to a maximal degree-2 path $P$ in $G$, we obtain a cycle, say, $C$, in $G'$ that does not contain any vertex from the set $X$ (since $X\subseteq V(G')$, $V(P)\cap X=\emptyset$). Next, note that if we replace the virtual vertex $v_{P}$ with $P$ in $C$, then $C$ is a cycle in $G$ that does not contain any vertex from $X$. This leads to a contradiction to the fact that $X$ is a feedback vertex set of $G$. Thus, $X$ is a feedback vertex set of $G'$ as well.
 
Since the application of either RR1 or RR2 reduces the number of vertices of $G$ by at least one (by the definitions of RR1 and RR2), it is clear that $|V(G')|<|V(G)|=t+1$. So, since $X$ is a feedback vertex set of $G'$, by the induction hypothesis, the lemma holds for $G'$. This implies that with probability at least $10^{-k}$, Algorithm \ref{algo22} outputs a virtual feedback vertex set, say, $\widehat{X'}$, of $G'$, such that $X$ is compatible with $\widehat{X'}$. Thus, $\widehat{X'}$ is a desired virtual feedback vertex set of $G$.
\medskip

\noindent \textbf{Subcase 2.2:} \bm{$X\nsubseteq V(G').$} Since $X$ is a minimal feedback vertex set of $G$ and $X\nsubseteq V(G')$, it is clear that $G'$ is obtained after applying RR2 (not RR1) on $G$. Let $P$ be the maximal degree-2 path in $G$ that has been replaced by a virtual vertex, say, $v_{P}$, by RR2. As $X\nsubseteq V(G')$, at least one of the vertices from $X$ must belong to $P$. Next, we claim that exactly one vertex from $X$, say, $x$, will belong to $P$. On the contrary, if there exist two distinct $x,y\in V(P)\cap X$, then $x$ and $y$ must belong to the same set of cycles in $G$ as they belong to the same maximal degree-2 path. This contradicts the fact that $X$ is minimal. Thus, $|V(P)\cap X|=1$. 

Next, by the definition of a safe set, note that $x$ belongs to the safe set corresponding to $v_{P}$. Let $X'=(X\setminus \{x\})\cup \{v_{P}\}$. Since the application of RR2 reduces the number of vertices of $G$ by at least one, it is clear that $|V(G')|\leq t$. So, by the induction hypothesis, Lemma \ref{lm2} holds for $G'$. This implies that with probability at least $10^{-k}$, Algorithm \ref{algo22} outputs a virtual feedback vertex set, say, $\widehat{X'}$, of $G'$ such that $X'$ is compatible with $\widehat{X'}$. Next, we claim that $X$ is compatible with $\widehat{X'}$ in $G$. Since $\widehat{X'}\setminus \{v_{p}\}$ is compatible with $X\setminus \{x\}$, and for $v_{p}\in \widehat{X'}$, we have $x\in X$ that belongs to the safe set of $v_{p}$, $\widehat{X'}$ is a desired virtual feedback vertex set of $G$. \qed
\end{proof}

Since the result is also true for the graph with $|V(G)|=t+1$, by the mathematical induction, the lemma holds. \qed
\end{proof}

\begin{definition} [Complement of a Matching]
If $M$ is a matching of a graph $G$, then $V(G)\setminus V_{M}$ is the \emph{complement} of $M$.
\end{definition}

By Proposition \ref{maxwt}, we have the following remark.
\begin{remark}
Let Algorithm $\mathcal{A}$ be any algorithm that solves \textsc{Max Weight Matching} in polynomial time.
\end{remark}

Next, consider Algorithm \ref{algo23}.
\medskip

	\begin{algorithm} [t]

		\KwIn{An instance $(G,\ell)$ of \textsc{Acyclic Matching} with $n=|V(G)|$;}
		\KwOut{An acyclic matching $M$ in $G$ of size at least $\ell$ or $\mathsf{No}$;}
	
	Call Algorithm \ref{algo22} with input $(G,n-2\ell)$;
	
	\If{$($Algorithm \ref{algo22} returns $\textsf{No})$}
	{\Return{$\mathsf{No}$;}}
	\ElseIf{$($Algorithm \ref{algo22} returns a virtual feedback vertex set $\widehat{X})$} 
	{
     Initialize $G_{W}=G$;
    
     For every $v_{P}\in \widehat{V}$, add a vertex $w_{P}$ to $G_{W}$ and make it adjacent to every vertex in the safe set corresponding to $v_{P}$. Call all edges introduced here \emph{new edges}.
    
     Remove all $v\in \widehat{X}\setminus \widehat{V}$ from $G_{W}$;
    
     Assign weight $c=|E(G)|+1$ to all new edges of $G_{W}$ and weight $1$ to all the remaining edges of $G_{W}$; 
    
     Call Algorithm $\mathcal{A}$ with input $(G_{W},\ell+|\widehat{V}|\cdot c)$;

	\If{$($Algorithm $\mathcal{A}$ returns a matching $M_{W}$ of weight at least $\ell+|\widehat{V}|\cdot c)$}{ $M=\{e\in M_{W}:$ weight of $e$ is $1\};$
		
	\Return{$M$;} }
\Else{ \Return{$\mathsf{No}$};}}
		\caption{}
	\label{algo23}	

	\end{algorithm}

\begin{remark}
Throughout this section, we call all the vertices and edges
introduced in Algorithm \ref{algo23} \emph{new vertices} and \emph{new edges}, respectively.
\end{remark}

For an illustrative example of Algorithm \ref{algo23}, consider the graphs shown in Figure \ref{1}. The construction of graph $G_{W}$ (defined in Algorithm \ref{algo23}) corresponding to graph $G$ (shown in Figure \ref{1} (i)) and the virtual feedback vertex set $\widehat{X}=\{v_{P_{3}},v_{P_{4}},v_{P_{6}}\}$ is shown in Figure \ref{algofigfig}. 

\begin{lem}\label{lmadd}
Let $G,$ $\ell$, $\widehat{X}$, $G_{W}$, $M_{W}$, and $M$ be as defined in Algorithm \ref{algo23}. If $M_{W}$ is of weight at least $\ell+|\widehat{V}|\cdot c$, then $M$ is an acyclic matching in $G$ of size at least $\ell$.
\end{lem}
\begin{proof}
Let $Y$ be the set of vertices in $G$ consisting of all $v\in \widehat{X}\setminus \widehat{V}$ and the set of all $M_{W}$-mates of the new vertices. Note that $|Y|\leq n-2\ell$ (as $|\widehat{X}|\leq n-2\ell$). Since the weight of $M_{W}$ is at least $\ell+|\widehat{V}|\cdot c$, $M_{W}$ must saturate all the new vertices of $G_{W}$. By the definition of $G_{W}$, all the new edges contribute exactly $|\widehat{V}|\cdot c$ weight to $M_{W}$. This implies that the remaining weight of $M_{W}$, which is at least $\ell$, must come from the edges of the graph $G-Y$. In turn, this further implies that at least $\ell$ edges (as the weight of each edge in $G-Y$ is 1) must form a matching, say, $M$, in the graph $G-Y$. Next, observe that $Y$ is compatible with $\widehat{X}$. So, due to Lemma \ref{lm1}, $Y$ is a feedback vertex set of $G$. Hence, $M$ is an acyclic matching in $G$ of size at least $\ell$.\qed
\end{proof}

\begin{lem} \label{lm3}
Let $(G,\ell)$ be a Yes-instance of \textsc{Acyclic Matching} with $n=|V(G)|$. Then, with probability at least $10^{2\ell-n}$, Algorithm \ref{algo23} returns an acyclic matching $M$ in $G$ of size at least $\ell$. 
\end{lem}
\begin{proof}
Since $(G,\ell)$ is a Yes-instance, there exists an acyclic matching, say, $M'$, of size at least $\ell$ in $G$. This implies that there exists a feedback vertex set, say, $X$, of size at most $n-2\ell$ in $G$ such that $X$ is the complement of $M'$. Let $k=n-2\ell$. By Lemma \ref{lm2}, if $G$ and $k$ are given as input, then with probability at least $10^{-k}$, Algorithm \ref{algo22} returns a virtual feedback vertex set $\widehat{X}$ such that $X$ is compatible with $\widehat{X}$. 

Next, note that if Algorithm \ref{algo22} returns a virtual feedback vertex set $\widehat{X}$, then Algorithm \ref{algo23} constructs an instance $(G_{W},\ell+|\widehat{V}|\cdot c)$ of \textsc{Max Weight Matching}. Now, we claim that $(G_{W},\ell+|\widehat{V}|\cdot c)$ is necessarily a Yes-instance. We also claim that if $M_{W}$ is a solution of $(G_{W},\ell+|\widehat{V}|\cdot c)$ returned by Algorithm $\mathcal{A}$, then $M$ (defined in Algorithm \ref{algo23}) is a solution of $(G,\ell)$. Since $\widehat{X}$ is compatible with $X$, each new vertex in $G_{W}$ is adjacent to at least one vertex in $X$, and any two distinct new vertices are adjacent to disjoint sets of vertices in $X$. So, we can define a matching, say, $M_{N}$, in $G_{W}$ such that each new vertex in $G_{W}$ is matched to some vertex in $X$. Note that the weight of $M_{N}$ must be exactly $|\widehat{V}|\cdot c$. Next, let $S$ be the set of vertices in $G$ consisting of all $v\in \widehat{X}\setminus \widehat{V}$ and the set of all $M_{N}$-mates of the new vertices. Note that $S\subseteq X$. Further, by the definition of $X$, the size of $G-X$ is at least $2\ell$, and $G-X$ has a perfect matching. Since $S\subseteq X$, the size of $G-S$ is at least $2\ell$, and it has a matching, say, $M_{S}$, of size at least $\ell$ (may not be perfect). As the weight of all edges in $G-S$ is $1$, the weight of $M_{S}$ must be at least $\ell$. By the definitions of $M_{N}$ and $M_{S}$, note that $\widehat{M}=M_{N}\cup M_{S}$ is a matching in $G_{W}$ of weight at least $\ell+|\widehat{V}|\cdot c$. So, $(G_{W},\ell+|\widehat{V}|\cdot c)$ is a Yes-instance. In turn, this means that the algorithm returns $M$. Then, by Lemma \ref{lmadd}, this set $M$ is an acyclic matching in $G$ of size at least $\ell$. \qed
\end{proof}

\begin{lem} \label{lm4}
Let $(G,\ell)$ be a No-instance of \textsc{Acyclic Matching} with $n=|V(G)|$. Then, with probability 1, Algorithm \ref{algo23} returns $\mathsf{No}$.
\end{lem}
\begin{proof}
Note that if  $(G,\ell)$ is a No-instance of \textsc{Acyclic Matching}, then there are two possibilities: $(i)$ there does not exist any feedback vertex set in $G$ of size at most $n-2\ell$, and $(ii)$ for every feedback vertex set $X$ of $G$ of size at most $n-2\ell$, $G-X$ does not have a perfect matching. 
 \begin{figure}[t]
 \centering
    \includegraphics[scale=0.9]{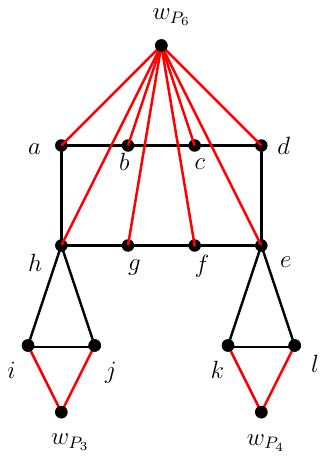}
    \caption{Note that the safe set corresponding to $v_{P_{3}}$ is $\{i,j\}$, corresponding to $v_{P_{4}}$ is $\{k,l\}$, and corresponding to $v_{P_{6}}$ is $\{a,b,c,d,e,f,g,h\}$. The edges shown with red color have a weight of $15$, and the edges shown with black color have a weight of $1$.}
    \label{algofigfig}
\end{figure}

If Algorithm \ref{algo22} returns $\mathsf{No}$, then we are done (as Algorithm \ref{algo23}, in this case, also returns $\textsf{No}$). So, assume that Algorithm \ref{algo22} returns a virtual feedback vertex set, say, $\widehat{X}$, of size at most $n-2\ell$. Next, we claim that the instance $(G_{W},\ell+|\widehat{V}|\cdot c)$ of the \textsc{Max Weight Matching} problem constructed by Algorithm \ref{algo23} is necessarily a No-instance. Targeting a contradiction, suppose that $(G_{W},\ell+|\widehat{V}|\cdot c)$ is a Yes-instance and $M_{W}$ is the matching returned by Algorithm $\mathcal{A}$ of weight at least $\ell+|\widehat{V}|\cdot c$.  By Lemma \ref{lmadd}, $M$ (defined in Algorithm \ref{algo23}) is an acyclic matching in $G$ of size at least $\ell$, a contradiction to the fact that $(G,\ell)$ is a No-instance. Thus, our assumption is wrong, and $(G_{W},\ell+|\widehat{V}|\cdot c)$ is a No-instance. Since this is reflected correctly by Algorithm \ref{algo23}, we conclude that if $(G,\ell)$ is a No-instance of \textsc{Acyclic Matching}, then, with probability 1, Algorithm \ref{algo23} returns $\mathsf{No}$. \qed
\end{proof}

We can improve the success probability of Algorithm \ref{algo22} and thus Algorithm \ref{algo23}, by repeating it, say, $t$
times, and returning a $\mathsf{No}$ only if we are not able to find a virtual feedback vertex set of size at most $k$ in each of
the repetitions. Clearly, due to Lemma \ref{lm4}, given a No-instance, even after repeating the procedure $t$ times, we will necessarily get $\mathsf{No}$ as an answer. However, given a Yes-instance, we return a $\mathsf{No}$ only if all $t$ repetitions
return an incorrect $\mathsf{No}$, which, by Lemma \ref{lm3}, has probability at most
\begin{equation} \label{eq}
(1-10^{-k})^{t}\leq (e^{-10^{-k}})^{t}\leq \frac{1}{e^{10^{-k}t}}.
\end{equation}

Note that we are using the identity $1+x\leq e^{x}$ in (\ref{eq}). In order to obtain a constant failure probability, we take $t=10^{k}$. By taking $t=10^{k}$, the success probability becomes at least $1-\frac{1}{e}.$ 

Thus, by the discussion above, we have the following theorem.
\AMBTG*

\subsection{Para-NP-hardness of \textsf{AMBIS}}\label{NP:hard}
In this section, we show that $\mathsf{AMBIS}$ is $\mathsf{NP}$-$\mathsf{hard}$ even for $k=0$. For this purpose, we first present how to construct an instance of \textsc{Acyclic Matching} from an instance of \textsc{Independent Set} \cite{garey}. The reduction that we use here is also given in \cite{panda1}, and therefore we give only the relevant details here.
\subsubsection{Construction}\label{panda:const}
Given a graph $G$, where $V(G)=\{v_{1},\ldots,v_{n}\}$, an instance of \textsc{Independent Set}, construct a graph $H$, an instance of \textsc{Acyclic Matching} as follows:
	\begin{itemize}
		\item[-] $V(H)=V(G) \cup V'$, where $V'=\{v'_{i}: v_{i}\in V(G)\}.$
		\item[-] $E(H)=E(G) \cup \{v_{i}v'_{i} : 1\leq i \leq n\} \cup \{v_{i}v'_{j} : v_{j}\in N_{G}(v_{i})\} \cup \{v'_{i}v'_{j} : v_{j}\in N_{G}(v_{i})\}.$
\end{itemize}  
See Figure \ref{fig4} for an illustration of the construction of $H$ from $G$.

\begin{figure}[t]
 \centering
    \includegraphics[scale=0.9]{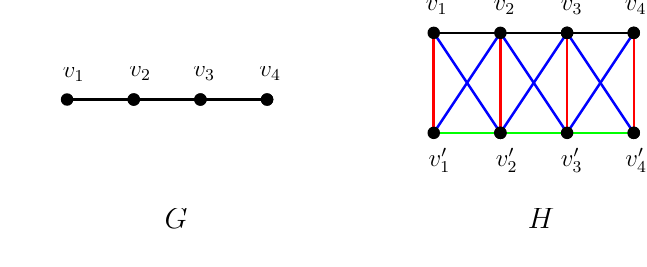}
    \caption{The construction of $H$ from $G$ in Construction \ref{panda:const}.}
    \label{fig4}
\end{figure}

Further, let us partition the edges of $H$ into the following four types:
\begin{enumerate}
	\item[-] Type-I=$\{v_{i}v'_{i}: v_{i}\in V(G)$ and $v'_{i}\in V'\}.$
	\item[-] Type-II= $\{v_{i}v_{j}: v_{i},v_{j}\in V(G) \}.$
	\item[-] Type-III= $\{v'_{i}v_{j}: v'_{i}\in V'$, $v_{j}\in V(G), i \neq j\}.$
	\item[-] Type-IV= $\{v'_{i}v'_{j}: v'_{i},v'_{j}\in V'\}.$
\end{enumerate}

\subsection{$\mathsf{NP}$-hardness Proof for $k=0$}

\begin{proposition} \label{S}
Let $G$ and $H$ be as defined in Construction \ref{panda:const}. Then, there exists a maximum acyclic matching $M$ in $H$ such that $M$ contains edges from $Type$-$I$ only. 
\end{proposition}

Based on Construction \ref{panda:const} and Proposition \ref{S}, we have the following corollary.

\begin{corollary} \label{ind1}
Let $G$ and $H$ be as defined in Construction \ref{panda:const}. Then, $\mathsf{IS}(H)=\mathsf{AM}(H)$.
\end{corollary}
\begin{proof}
By Proposition \ref{S}, there exists a maximum acyclic matching, say, $M$, in $H$ such that $M$ contains only Type-I edges. Without loss of generality, let $M=\{v_{1}v'_{1}, \ldots, v_{\ell}v'_{\ell}\}$. Define a set $I=\{v_{1}, \ldots, v_{\ell}\}$ in $H$. First, we claim that $I$ is an independent set in $H$. Else, if there exist distinct $v_{i},v_{j}\in I$ such that $v_{i}v_{j}\in E(H)$, then by the definition of $H$, $v'_{i}v'_{j}\in E(H)$. This implies that $v_{i},v_{j},v'_{j},v'_{i},v_{i}$ forms a cycle in $G[V_{M}]$, a contradiction to the fact that $M$ is an acyclic matching in $H$. Thus, $I$ is an independent set in $H$. Next, we claim that $I$ is maximum in $H$. For the sake of contradiction, assume that $I'$ is a maximum independent set of $H$ and $|I'|>|I|$. Note that if $v\in I'$, then $v'\notin I'$ (as $vv'\in E(H)$), and vice versa. Define $M'=\{vv': v'\in I'$ or $v\in I'\}$. Since $I'$ is an independent set, $M'$ is an acyclic matching of $H$ and $|M'|>|M|$, a contradiction to the fact that $M$ is a maximum acyclic matching of $H$. Thus, $\mathsf{IS}(H)=\mathsf{AM}(H)$.
\end{proof}

Next, consider the following result.

\begin{proposition} [\cite{panda1}] \label{ind}
Let $G$ and $H$ be as defined in Construction \ref{panda:const}. Then, $G$ has an independent set of size at least $\ell$ if and only if $H$ has an acyclic
matching of size at least $\ell$.
\end{proposition}

By Proposition \ref{ind} and Corollary \ref{ind1}, it is clear that $\mathsf{AMBIS}$ is $\mathsf{NP}$-$\mathsf{hard}$ for $k=0$, and thus we have the following result.
\AMBIS*

\section{Negative Kernelization Results}\label{secnopoly}

\subsection{Vertex Cover Number}
In this section, we prove that under the assumption that $\mathsf{NP} \nsubseteq \mathsf{coNP} \slash \mathsf{poly}$, there does not exist
any polynomial kernel for \textsc{Acyclic Matching} when parameterized by the vertex cover number of the input graph plus the size of the matching. For this purpose, we give an OR-cross-composition (see Section \ref{PC}) from \textsc{Exact-3-Cover} (defined below), which is $\mathsf{NP}$-$\mathsf{hard}$ \cite{garey}:\\

\noindent\fbox{ \parbox{145mm}{
		\noindent\underline{\textsc{Exact-3-Cover}}:
		
		\smallskip
		\noindent\textbf{Instance:} A set $X$ with $|X|=3c$, where $c\in \mathbb{N}$, and a collection $\mathcal{S}$ of 3-element subsets of $X$.
		
		\noindent\textbf{Question:} Does there exist a subcollection $\mathcal{S'}$ of $\mathcal{S}$ such that every element of $X$ appears in exactly one member of $\mathcal{S'}$?}}
\bigskip

We remark that our reduction is inspired by the reduction given by Gomes et al. \cite{gomes} to prove that
there is no polynomial kernel for \textsc{Induced Matching} when parameterized by the vertex cover number plus the size of the matching of the input graph unless $\mathsf{NP}\subseteq \mathsf{coNP} \slash \mathsf{poly}$.

\subsubsection{Construction} \label{const:kernel}
Let $\mathcal{I}=\{(X_{1},\mathcal{S}_{1}),\ldots,(X_{t},\mathcal{S}_{t})\}$ be a collection of $t$ instances of \textsc{Exact-3-Cover}. Without loss of generality, let $X_{i}=X=[n]$ and $|\mathcal{S}_{i}|=m$ for all $i\in[t]$ (note that $n=3c$ for some $c\in \mathbb{N}$). Let $\mathcal{C}=\displaystyle{\bigcup_{i\in[t]}\mathcal{S}_{i}}$, and assume that $\mathcal{S}_{i}\neq \mathcal{S}_{j}$ for every distinct $i,j\in[t]$. Also, we denote $\mathcal{C}$ as $\{s_{1},s_{2},\ldots, s_{|\mathcal{C}|}\}$. Furthermore, we denote by $(G,\ell)$ the instance of \textsc{Acyclic Matching} that we construct in this section.  First, we introduce the vertex set $X'=\{v_{a}:a\in X\}$ to $G$. Next, we define the set gadgets as follows.
\medskip

\noindent \textbf{Set Gadget:} For each $s_{j}\in \mathcal{C}$, where $s_{j}=\{a,b,c\}$, we add a copy of $K_{2,3}$ with vertices $\{u_{ja},u_{jb},u_{jc}\}$ in one partition and $\{u_{j},w_{j}\}$ in the other, and introduce an edge between $u_{j}$ and $w_{j}$. We also introduce two pendant edges $u_{j}u'_{j}$ and $w_{j}w'_{j}$ incident on $u_{j}$ and $w_{j}$, respectively. Let us refer to the set gadget corresponding to $s_{j}\in \mathcal{C}$ as $Q_{j}$. See Figure \ref{fig1} for an illustration. Furthermore, for every $Q_{j}$, if $s_{j}=\{a,b,c\}$, then we call the vertices $\{u_{ja},u_{jb},u_{jc}\}$ \emph{interface vertices}. For every $s_{j}\in \mathcal{C}$, where $s_{j}=\{a,b,c\}$ and $v_{d}\in X'$, $u_{ja}v_{d}\in E(G)$ if and only if $a=d$.
\medskip

\begin{figure}[t]
 \centering
    \includegraphics[scale=0.9]{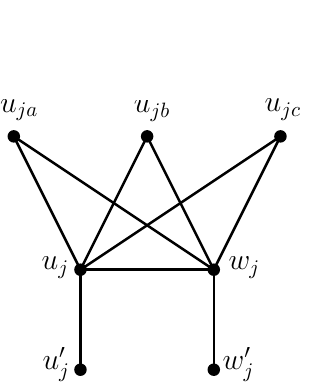}
    \caption{The construction of a set gadget $Q_{j}$ corresponding to $s_{j}=\{a,b,c\}\in  \mathcal{C}$.}
    \label{fig1}
\end{figure}

\noindent \textbf{Instance Selector:} Introduce a $K_{1,t}$ with $p$ as the central vertex and $p_{i}$, $i\in [t]$, as leaves. Let $P=\{p_{i}:i\in [t]\}$. For each $p_{i}, i\in [t]$ and $s_{j}\in \mathcal{C}\setminus \mathcal{S}_{i}$, introduce edges between $p_{i}$ and the interface vertices of the set gadget $Q_{j}$. See Figure \ref{fig2} for an illustration of the construction of $G$.

Finally, we set $\ell=2|\mathcal{C}|+\frac{2n}{3}+1$.

Now, consider the following lemma.
\begin{lem} \label{sizevc}
Let $G$ be as defined in Construction \ref{const:kernel}. Then, $G$ has a vertex cover of size $\mathcal{O}(n^{3})$.
\end{lem}
\begin{proof}
As $P \cup X'$ is an independent set of $G$, it can be observed that the set $Y=V(G)\setminus (P \cup X')$ is a vertex cover of $G$. Furthermore, note that $|Y|=7|\mathcal{C}|+1$. Since all the elements in $\mathcal{C}$ are distinct, there can be at most $\binom{n}{3}$ elements in $\mathcal{C}$. Thus, we have $|Y|\in \mathcal{O}(n^{3})$. \qed
\end{proof}
\subsubsection{From Exact-3-Cover to Acyclic Matching}

We remark that by saying that $(X,\mathcal{S})$ admits a solution of $\mathcal{I}$, we mean that there exists a set $\mathcal{S'}\subseteq \mathcal{S}$ such that if $\mathcal{S'}=\{s'_{1},\ldots,s'_{\frac{n}{3}}\}$, then $\bigcup^{\frac{n}{3}}_{i=1}$ $s'_{i}=X$.

Now, consider the following lemma.
\begin{figure}[t]
    \centering
    \includegraphics[scale=0.78]{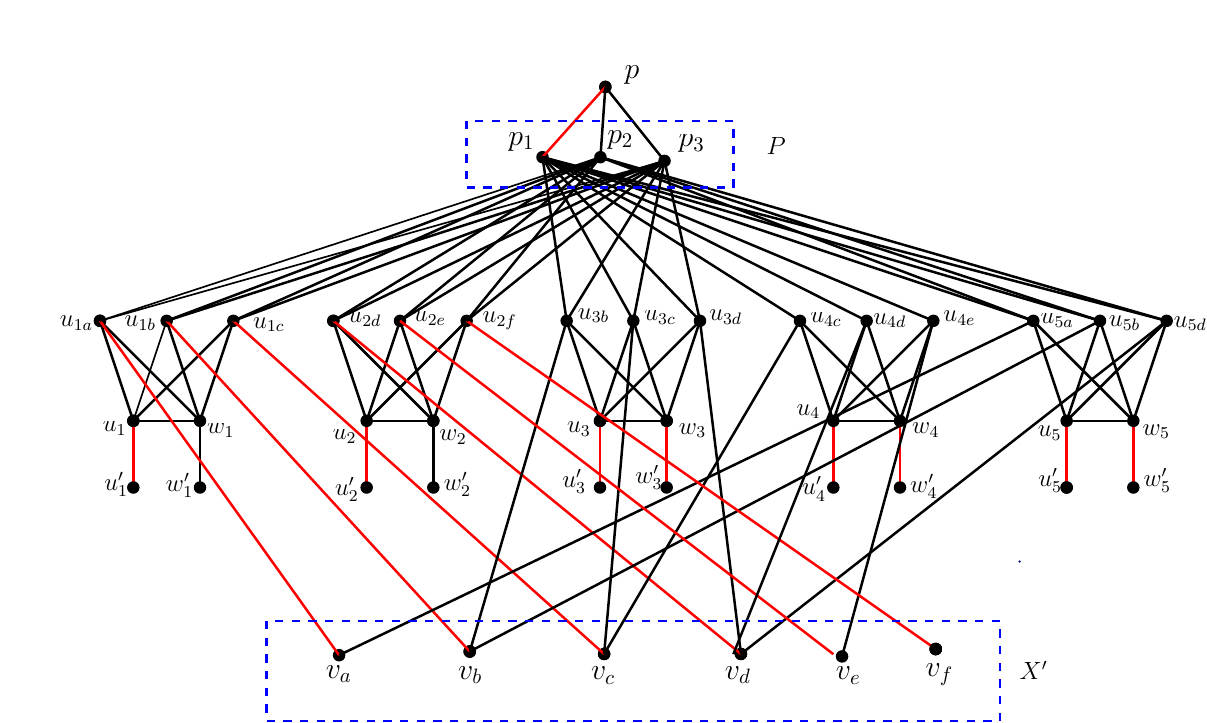}
    \caption{The construction of $G$ from $\mathcal{I}=\{(X,\mathcal{S}_{1}),(X,\mathcal{S}_{2}),(X,\mathcal{S}_{3})\}$, where $X=\{a,b,c,d,e,f\}$, $\mathcal{S}_{1}=\{\{a,b,c\},\{d,e,f\}\},\mathcal{S}_{2}=\{\{b,c,d\},\{c,d,e\}\},$ and $\mathcal{S}_{3}=\{\{a,b,d\},\{c,d,e\}\}.$ The edges shown with red color represent an acyclic matching in $G$.}
    \label{fig2}
\end{figure}
\begin{lem} \label{AM-1}
If $(X,\mathcal{S}_{q})$ admits a solution of $\mathcal{I}$, then $G$ admits an acyclic matching of size $\ell$.
\end{lem}
\begin{proof}
Let $\mathcal{S'}$ be a solution of $(X,\mathcal{S}_{q})$. For $s_{j}\in \mathcal{S'}$, if $s_{j}=\{a,b,c\}$, then add the edges $\{u_{ja}v_{a},u_{jb}v_{b},u_{jc}v_{c},u_{j}u'_{j}\}$ to $M$. For $s_{k}\in \mathcal{C} \setminus \mathcal{S'}$, add edges $\{u_{k}u'_{k},w_{k}w'_{k}\}$ to $M$. Finally, add the edge $p_{q}p$ to $M$. Since $|\mathcal{S'}|=\frac{n}{3}$, it is easy to see that $|M|=\frac{4n}{3}+2(|\mathcal{C}|-\frac{n}{3})+1=\ell$. Now, it remains to show that $M$ is an acyclic matching in $G$. 

For the sake of contradiction, assume that $G[V_{M}]$ contains a cycle, say $C$. First, we claim that none of the vertices from $X'$ belong to $C$. Since $d(v)=1$ in $G[V_{M}]$ for every $v\in X'$, our claim holds.
The same is true for the vertex $p$. Next, by the definition of $G$ and $M$, it is clear that $p_{q}$ is incident to only the unsaturated interface vertices. Therefore, we say that $d(p_{q})=1$ in $G[V_{M}]$. By the arguments presented above, we say that the cycle $C$ must be contained entirely in $Q_{j}$ corresponding to some $s_{j}\in \mathcal{C}$. For $s_{j}\in \mathcal{S'}$, if $s_{j}=\{a,b,c\}$, then the subgraph induced by the $M$-saturated vertices restricted to $Q_{j}$ is a $K_{1,4}$ with the partition $\{u'_{j},u_{ja},u_{jb},u_{jc}\}$ and $\{u_{j}\}$. For $s_{k}\notin \mathcal{S'}$, if $s_{k}=\{a,b,c\}$, then the subgraph induced by the $M$-saturated vertices restricted to $Q_{k}$ is a $P_{4}$, namely, $u'_{k},u_{k},w_{k},w'_{k}$. In this way, we deny the possibility of the existence of $C$ in $G[V_{M}]$. Hence, $M$ is an acyclic matching in $G$. \qed
\end{proof}

\subsubsection{From Acyclic Matching to Exact-3-Cover}
Consider graph $G$ as defined in Construction \ref{const:kernel}. Let us call the edges between $X'$ and the set of interface vertices \emph{cross edges}. From the definition of matching, it is clear that in any matching $M$ of $G$, at most $n$ cross edges belong to $M$ (as $|X'|=n$). Furthermore, let us call the edges between the set of interface vertices and $P$ as \emph{upper edges}. Now, consider the following definitions that will be used in proving Lemmas \ref{am} and \ref{AM-2}.

\begin{definition} [Happy Gadget]
A set gadget is \emph{happy} with respect to a matching $M$ if at least one of its interface vertex is saturated by $M$ through a cross edge.
\end{definition}

\begin{definition} [Touched Gadget]
A set gadget is \emph{touched} with respect to a matching $M$ if at least one of its interface vertex is saturated by $M$ through an upper edge.
\end{definition}

\begin{lem} \label{am}
In every solution $M$ of $(G,\ell)$, $p$ is saturated by $M$.
\end{lem}
\begin{proof}
Let $M$ be an arbitrary but fixed solution of $(G,\ell)$.  First, assume that $p\notin V_{M}$ and $P\cap V_{M}=\emptyset$. Next, observe that a happy gadget, say, $Q_{j}$, contributes at most one edge to $M$ that lies entirely within $Q_{j}$ (i.e., $|M\cap E(Q_{j})|\leq 1$, if $Q_{j}$ is happy), and a set gadget that is not happy, say, $Q_{k}$, contributes at most two edges to $M$ that lie entirely within $Q_{k}$ (i.e., $|M\cap E(Q_{k})|\leq 2$, if $Q_{k}$ is not happy). If none of the set gadgets is happy, then $|M| \leq 2|\mathcal{C}|<\ell$, a contradiction. Thus, some of the set gadgets must be happy. 

Next, let us assume that $0<n'\leq n$ set gadgets are happy. If $n'<\frac{n}{3}$, then $|M|\leq n'+3n'+2(|\mathcal{C}|-n')=2n'+2|\mathcal{C}|< \frac{2n}{3}+2|\mathcal{C}|<\ell$, a contradiction. On the other hand, if $n'\geq \frac{n}{3}$, then, first, recall that at most $n$ cross edges belong to $M$. Therefore, $|M|\leq n'+n+2(|\mathcal{C}|-n')=n-n'+2|\mathcal{C}|\leq n-\frac{n}{3}+2|\mathcal{C}|<\ell$, a contradiction. Thus, $P\cap V_{M}\neq\emptyset$.

Next, without loss of generality, assume that $P'=\{p
_{1},\ldots,p
_{k}\}\in V_{M}$. Since
$p\notin V_{M}$, it is easy to note that every $p_{i}\in P'$ must be matched in $M$ via an upper edge. Also, note that distinct $p_{i}'$s \emph{touch} distinct set gadgets. Else, if two distinct $p_{i}$ and $p_{m}$ in $P'$ touch the same set gadget, say, $Q_{j}$, where $s_{j}=\{a,b,c\}$, then, assuming $p_{i}u_{ja},p_{m}u_{jb}\in M$, $p_{i},u_{ja},p_{m},u_{jb},p_{i}$ forms a cycle in $G[V_{M}]$, which is a contradiction. Now, observe that every touched set gadget, say, $Q_{j}$, is either happy or satisfies the condition that $|E(Q_{j})\cap M|=1$. Furthermore, if $Q_{j}$ is happy, then $|E(Q_{j})\cap M|=0$, and in this case, we can replace the upper edge touching $Q_{j}$ with $w_{j}w'_{j}$. So, without loss of generality, we can assume that every touched set gadget is not happy.
Let us now assume that $0<n'\leq n$ set gadgets are happy and $k$ set gadgets are touched. If $n'<\frac{n}{3}$, then $|M|\leq n'+3n'+2k+2(|\mathcal{C}|-n'-k)=2n'+2|\mathcal{C}|< \frac{2n}{3}+2|\mathcal{C}|<\ell$, a contradiction. On the other hand, if $n'\geq \frac{n}{3}$, then $|M|\leq n'+n+2(|\mathcal{C}|-n'-k)+2k=n-n'+2|\mathcal{C}|\leq n-\frac{n}{3}+2|\mathcal{C}|<\ell$, a contradiction. Since we get a contradiction in both cases, $ M$ must saturate $p$. \qed
\end{proof} 

From Lemma \ref{am}, we have the following corollary.
\begin{corollary}\label{corr:AM}
There is an edge of the form $p_{q}p$ for some $q\in [t]$
in every solution of $(G,\ell)$.
\end{corollary}
\begin{lem} \label{AM-2}
If $(G,\ell)$ admits a solution, then at least one instance $(X,\mathcal{S}_{q})\in \mathcal{I} $ also admits a solution.
\end{lem}
\begin{proof}
By Corollary \ref{corr:AM}, let $M$ be a solution of $(G,\ell)$ with $pp_{q}\in M$ for some $q\in [t]$. Let $\mathcal{Q}$ be the set of happy gadgets in $G$ and let there be $k$ touched set gadgets in $G$. We can assume without loss of generality that every touched set gadget is not happy. We claim that $|\mathcal{Q}|=\frac{n}{3}$. If $n'=|\mathcal{Q}|<\frac{n}{3}$, then $|M|\leq n'+3n'+2k+2(|\mathcal{C}|-n'-k)+1=2n'+2|\mathcal{C}|+1< \frac{2n}{3}+2|\mathcal{C}|+1=\ell$, a contradiction. If $n'=|\mathcal{Q}|>\frac{n}{3}$, then $|M|\leq n'+2k+n+2(|\mathcal{C}|-n'-k)+1=n-n'+2|\mathcal{C}|+1<n-\frac{n}{3}+2|\mathcal{C}|+1=\ell$, a contradiction. Thus, $|\mathcal{Q}|=\frac{n}{3}$. 

We know that for every $Q_{j}\in \mathcal{Q}$, $|E(Q_{j})\cap M|\leq1$, and for every $Q_{k}\notin \mathcal{Q}$, $|E(Q_{k})\cap M|\leq 2$. Next, we claim that $|E(Q_{j})\cap M|=1$ for every $Q_{j}\in \mathcal{Q}$, and $|E(Q_{k})\cap M|=2$ for every $Q_{k}\notin \mathcal{Q}$. On the contrary, if we assume that either $E(Q_{j})\cap M=\emptyset$ for some $Q_{j}\in \mathcal{Q}$ or $|E(Q_{k})\cap M|\leq 1$ for some $Q_{k}\notin \mathcal{Q}$, then by the similar arguments as presented in the first part of the proof, we say that $|M|<\ell$, a contradiction.  With this information, it is easy to note that to justify the size of $M$, exactly $n$ cross edges must belong to $M$.

Next, we claim that the interface vertices corresponding to happy gadgets must not be adjacent to $p_{q}$. If $Q_{j}$ is an arbitrary happy gadget, where $s_{j}=\{a,b,c\}$, then if $p_{q}u_{ja},p_{q}u_{jb},p_{q}u_{jc}\in E(G)$, then a cycle $p_{q},u_{ja}, u_{j},u_{jb}$ (or $p_{q},u_{ja}, w_{j},u_{jb}$) will be formed in $G[V_{M}]$, a contradiction. It implies that, for each $Q_{j}\in \mathcal{Q}$, we have that $s_{j}\in \mathcal{S}_{q}$. Since vertices $\{a,b,c\}$ of $X'$ are matched with vertices $\{u_{ja},u_{jb},u_{jc}\}$ of $Q_{j}\in \mathcal{Q}$, it follows that $\{s_{j}: Q_{j}\in \mathcal{Q}\}$ is a solution
to $(X,\mathcal{S}_{q})$. \qed
\end{proof}

By Proposition \ref{orcross} and Lemmas \ref{sizevc} \ref{AM-1}, and \ref{AM-2}, we have the following theorem.
\NOPOLYVC*

\subsection{Vertex Deletion Distance to Clique}
Observe that in Construction \ref{const:kernel}, if we make the set $P$ a clique and proceed exactly as before, then with only minor changes (specified below), we can show that \textsc{Acyclic Matching} does not admit a polynomial kernel when parameterized by the vertex deletion distance to a clique. Let $P'=P\cup \{p\}$.

\begin{lem} \label{vddtc}
Let $G$ be as defined in Construction \ref{const:kernel} with the additional condition that $P$ is a clique. Then, $G$ has a clique modulator of size $\mathcal{O}(n^{3})$.
\end{lem}
\begin{proof}
As $Y=V(G)\setminus P'$ is a clique modulator of $G$, $|Y|=7|\mathcal{C}|+n$, and $|\mathcal{C}|\in \mathcal{O}(n^{3})$, we have $|Y|\in \mathcal{O}(n^{3})$. \qed
\end{proof}
\begin{remark}
If we make the set $P$ a clique in Construction \ref{const:kernel}, then Lemma \ref{sizevc} may not hold.
\end{remark}
\begin{lem} \label{gk}
If $(G,\ell)$ admits a solution, then at least one instance in $\mathcal{I} $ also admits a solution.
\end{lem}
\begin{proof}
 Let $M$ be a solution to $(G,\ell)$. First, observe that $|V_{M}\cap P'|=2$, else $G[V_{M}]$ will contain a cycle. If $M$ saturates $p$, then the arguments given in the proof of Lemma \ref{AM-2} can be used to show the existence of a solution of $\mathcal{I}$ (note that in this case, there will be no touched gadgets). So, we assume that $M$ does not saturate $p$. Next, we claim that $M\cap E(P')\neq \emptyset$. Note that if $p_{i}$ for some $i\in [t]$ is matched to a vertex
outside of $P$, then the arguments given in the proof of Lemma \ref{am} (second paragraph) can be used to show that $|M|<\ell$, a contradiction. Thus, $M\cap E(P') \neq\emptyset$.  Now, it is clear that $M$ picks an edge of the form $p_{i}p_{k}$ for some $i,k\in [t]$. If $Q_{j}$ is an arbitrary happy gadget, where $s_{j}=\{a,b,c\}$, then we claim that $u_{ja},u_{jb},u_{jc}$ are not adjacent to $p_{i}$ and $p_{k}$. For the sake of contradiction, without loss of generality, assume that $p_{i}u_{ja}\in E(G)$. By the definition of $G$, $p_{i}u_{jb},p_{i}u_{jc}\in E(G)$. It implies that $p_{i},u_{ja}, u_{j},u_{jb}$ or $p_{i},u_{ja}, w_{j},u_{jb}$ is a cycle in $G[V_{M}]$, depending on whether $u_{j}u'_{j}$ or $w_{j}w'_{j}\in M$, respectively (note that to justify the size of $M$, either $u_{j}u'_{j}$ or $w_{j}w'_{j}$ should belong to $M$). It leads to a contradiction to the fact that $M$ is an acyclic matching. Thus, for each $Q_{j}\in \mathcal{Q}$, we have that $s_{j}\in \mathcal{S}_{i},\mathcal{S}_{k}$. Since vertices $\{a,b,c\}$ of $X'$ are matched with vertices $\{u_{ja},u_{jb},u_{jc}\}$ of $Q_{j}\in \mathcal{Q}$, it follows that $\{s_{j}: Q_{j}\in \mathcal{Q}\}$ is a solution
to $(X_{i},\mathcal{S}_{i})$ as well as $(X_{k},\mathcal{S}_{k})$. \qed
\end{proof}  

By Proposition \ref{orcross} and Lemmas \ref{AM-1}, \ref{vddtc}, and \ref{gk}, we have the following theorem.
\NOPOLYDTC*

\section{Conclusion and Future Research} \label{conclu}

Moser and Sikdar \cite{moser} showed that \textsc{Induced Matching} for planar graphs admits a kernel of size $\mathcal{O}(\mathsf{IM(G)})$. The kernelization technique used in \cite{moser} has two main components, viz., data reduction rules and an intrinsic property of a maximum induced matching - let us call it \emph{Property $\mathsf{P}$}. For completeness, we give the reduction rules and Property $\mathsf{P}$ below.
\medskip

\noindent \textbf{Reduction Rules:}
\begin{enumerate}
    \item[(R0)] Delete vertices of degree 0.
\item[(R1)] If a vertex $u$ has two distinct neighbors $x$, $y$ of degree $1$, then delete $x$.
\item[(R2)] If $u$ and $v$ are two vertices such that $|N(u)\cap N(v)|\geq 2$ and if there exist $x,y\in N(u)\cap N(v)$ with
$d(x) = d(y) = 2$, then delete $x$.
\end{enumerate}

\noindent \textbf{Property $\mathsf{P}$:} If $M$ is a maximum induced
matching in a graph $G$, then for each vertex $v\in V(G)$, there exists a vertex $u\in V_{M}$ such that $\widehat{d}(u,v)\leq 2$.
\medskip

We note that \textsc{Acyclic Matching} and \textsc{Uniquely Restricted Matching} also admit kernels of sizes $\mathcal{O}(\mathsf{AM(G)})$ and $\mathcal{O}(\mathsf{URM(G)})$, respectively, on planar graphs, as the data reduction rules R0-R2 and Property $\mathsf{P}$ hold true for these problems as well.  
On similar lines, based on the result given in \cite{moser} that states that \textsc{Induced Matching} admits a quadratic kernel (with respect to the maximum degree of the input graph) for bounded degree graphs, we note that there exists a quadratic kernel with respect to the maximum degree of the input graph for both \textsc{Acyclic Matching} and \textsc{Uniquely Restricted Matching} for bounded degree graphs. In fact, for \textsc{Uniquely Restricted Matching}, we further note that the quadratic kernel can be improved to a linear kernel by stating a stronger property than Property $\mathsf{P}$. The following is true for any maximum uniquely restricted matching.
\medskip

\noindent \textbf{Property $\mathsf{\widehat{P}}$:} If $M$ is a maximum uniquely restricted matching in a graph $G$, then for each vertex $v\in V(G)$, there exists a vertex $u\in V_{M}$ such that $\widehat{d}(u,v)\leq 1$. \medskip

\noindent \begin{proof}
Targeting a contradiction, let there exists a vertex, say, $v\in V(G)$, such that for all $u\in V_{M}$, $\widehat{d}(u,v)\geq 2$. Now, for some $w\in N(v)$, define $M'=M\cup \{vw\}$. Next, we claim that $M'$ is a uniquely restricted matching in $G$. By Proposition \ref{defurm}, $M'$ is a uniquely restricted matching in $G$ if and only if there does not exist any alternating cycle in $G$. Note that if there exists an alternating cycle in $G$, then it must contain the edge $vw$, else it contradicts the fact that $M$ is a uniquely restricted matching in $G$. Observe that if $vw$ belongs to an alternating cycle in $G$, then at least one neighbor of $v$ other than $w$ must be saturated by $V_{M'}$ (and hence by $V_M$), which is not possible. \qed
\end{proof}

A natural question that often arises in Parameterized Complexity, whenever a problem $\mathrm{\Pi}$ is $\mathsf{FPT}$ with respect to a parameter $k$, is whether $\mathrm{\Pi}$ is $\mathsf{FPT}$ for a parameter smaller than $k$ or not. One possible direction for future research is to seek a below-guarantee parameter smaller than the parameter $\frac{n}{2}-\ell$, so that \textsc{Acyclic Matching} remains $\mathsf{FPT}$. Also, it would be interesting to see if the running time in Theorem \ref{ambt1} can be substantially improved. Apart from that, we strongly believe that the arguments presented in this work (in Section \ref{FPT:AMBTG}) will be useful for other future works concerning problems where one seeks a solution that, among other properties, satisfies that it is itself, or its complement, a feedback vertex set, or, much more generally, an alpha-cover (see \cite{fomin}).

\begin{center}\textbf{Acknowledgments}\end{center}
 The authors are supported by the European Research Council (ERC) project titled PARAPATH (101039913).
 


\end{document}